\documentclass[journal]{IEEEtran}
\usepackage{graphicx}
\usepackage[cmex10]{amsmath}
\usepackage{cases}
\usepackage[tight,footnotesize]{subfigure}
\usepackage{amsthm} 
\usepackage{amssymb}
\usepackage{algorithm}
\usepackage{algorithmic}
\usepackage{multirow}
\usepackage{amsmath}
\usepackage{xcolor}
\usepackage{CJK}
\usepackage{subeqnarray}
\usepackage{cases}
\usepackage{enumerate}
\usepackage{cuted}
\usepackage{booktabs}
\usepackage{mathrsfs}
\usepackage{bbm}
\usepackage{array}
\newcolumntype{M}[1]{>{\centering\arraybackslash}m{#1}}
\usepackage[colorlinks]{hyperref}

\newtheorem{theorem}{\hskip\parindent \it Theorem}
\newtheorem{remark}{\hskip\parindent\it{Remark}}
\newtheorem{corollary}{\hskip\parindent\it{Corollary}}
\ifCLASSINFOpdf
\else
\fi

\begin{document}

\title{Fluid Reconfigurable Intelligent Surface\\ with Element-Level Pattern Reconfigurability: Beamforming and Pattern Co-Design}

\author{Han Xiao,~\IEEEmembership{Graduate Student Member,~IEEE,} 
            Xiaoyan Hu,~\IEEEmembership{Member,~IEEE,} 
            Kai-Kit~Wong,~\IEEEmembership{Fellow,~IEEE}\\
            Xusheng Zhu,~\IEEEmembership{Member,~IEEE}, 
            Hanjiang Hong,~\IEEEmembership{Member,~IEEE},  
            Chan-Byoung Chae,~\IEEEmembership{Fellow,~IEEE}
\vspace{-8mm}

\thanks{H. Xiao and X. Hu are with the School of Information and Communications Engineering, Xi'an Jiaotong University, Xi'an 710049, China (e-mail: $\{\rm hanxiaonuli@stu.xjtu.edu.cn; xiaoyanhu@xjtu.edu.cn\}$).}
\thanks{K. Wong, X. Zhu and H. Hong are with the Department of Electronic and Electrical Engineering, University College London, London WC1E7JE, U.K.~(e-mail: $\rm \{kai\text{-}kit.wong,xusheng.zhu,hanjiang.hong\}@ucl.ac.uk$). K. Wong is also affiliated with Yonsei Frontier Laboratory, Yonsei University, Seoul, 03722, Republic of Korea.}
\thanks{C. B. Chae is with School of Integrated Technology, Yonsei University, Seoul, 03722, Republic of Korea (e-mail: $\rm cbchae@yonsei.ac.kr$).}
}
\maketitle

\begin{abstract}
This paper proposes a novel pattern-reconfigurable fluid reconfigurable intelligent surface (FRIS) framework, where each fluid element can dynamically adjust its radiation pattern based on instantaneous channel conditions. To evaluate its potential, we first conduct a comparative analysis of the received signal power in point-to-point communication systems assisted by three types of surfaces: (1) the proposed pattern-reconfigurable FRIS, (2) a position-reconfigurable FRIS, and (3) a conventional RIS. Theoretical results demonstrate that the pattern-reconfigurable FRIS provides a significant advantage in modulating transmission signals compared to the other two configurations. To further study its capabilities, we extend the framework to a multiuser communication scenario. In this context, the spherical harmonics orthogonal decomposition (SHOD) method is employed to accurately model the radiation patterns of individual fluid elements, making the pattern design process more tractable. An optimization problem is then formulated with the objective of maximizing the weighted sum rate among users by jointly designing the active beamforming vectors and the spherical harmonics coefficients, subject to both transmit power and pattern energy constraints. To tackle the resulting non-convex optimization problem, we propose an iterative algorithm that alternates between a minimum mean-square error (MMSE) approach for active beamforming and a Riemannian conjugate gradient (RCG) method for updating the spherical harmonics coefficients. Simulation results show that the proposed pattern-reconfigurable FRIS significantly outperforms traditional RIS architectures based on the 3GPP 38.901 and isotropic radiation models, achieving average performance gains of $161.5\%$ and $176.2\%$, respectively. Additionally, it reduces the required number of antennas and RIS elements by over $300\%$, offering substantial improvements in hardware efficiency.
\end{abstract}

\IEEEpeerreviewmaketitle

\vspace{-2mm}
\section{Introduction}\label{sec:S1}
\IEEEPARstart{R}{econfigurable} intelligent surface (RIS) technology has recently attracted considerable attention in the wireless communications community due to its ability to dynamically manipulate the electromagnetic environment to control wireless signal propagation \cite{huang2019reconfigurable,Basar-2019}. By doing so, RIS enables the creation of an end-to-end, controllable virtual channel between transceivers, which can significantly enhance system coverage---particularly in scenarios where direct line-of-sight (LoS) links are severely attenuated or blocked, a challenge that becomes increasingly prevalent at higher carrier frequencies \cite{Zhou-RIS2021}. Owing to these capabilities, RIS has been widely explored across a range of wireless network applications to boost overall system performance. These include secure communications \cite{li2025covert, liu2025ris, xiao2024simultaneously, xiao2025robust}, integrated sensing and communication (ISAC) \cite{liu2025exploiting, xu2024intelligent, du2025intelligent}, mobile edge computing (MEC) \cite{huIRS2023,xiao2025star-ris, xiao2025energy, hu2021reconfigurable}, cell-free network architectures \cite{chen2024distributed}, among others.

Despite its promise, the practical implementation of RIS faces several critical challenges. One major issue lies in the complexity of channel estimation between transceivers \cite{jian2022reconfigurable}. Because RIS operates passively, achieving high system performance often requires deploying many reflecting elements. However, increasing the density of these elements dramatically raises the pilot overhead and adds substantial complexity to both system design and the channel estimation process, making real-time operation and scalability difficult. Another significant challenge is the multiplicative fading effect inherent in RIS-assisted systems \cite{zhi2022active}. In such systems, the overall path loss of the transmitter-RIS-receiver link becomes the product of the individual losses of the transmitter-to-RIS and RIS-to-receiver channels. This multiplicative nature leads to severe signal attenuation, which can substantially reduce the effectiveness of RIS in improving communication quality, particularly in environments with limited/obstructed LoS or high propagation losses. To overcome these limitations, a promising direction is to introduce new degrees of freedom (DoFs) into RIS-assisted systems, thereby enhancing their flexibility and adaptability---without significantly increasing hardware complexity.

\vspace{-2mm}
\subsection{Related Works}
Against this backdrop, originally proposed by Wong {\em et al.}~in \cite{I22_wong2020perflim,wong2021fluid}, the fluid antenna system (FAS) has emerged as a promising solution. FAS represents a new class of reconfigurable antennas \cite{I27_basbug2017design,I24_shen2024design,zhang2025novel,Liu-2025arxiv} capable of dynamically adjusting both their position and shape, thereby introducing additional DoFs at the physical layer \cite{new2024tutorial,Hong-2025arxiv, Lu-2025}. In recent years, considerable progress has been made in exploring the potential of FAS in various scenarios. For instance, \cite{new2024information} analyzed the diversity-multiplexing trade-off in FAS-enabled multiple-input multiple-output (MIMO) systems, highlighting the system's ability to flexibly balance reliability and throughput. In another line of work, FAS has been integrated into massive multiple access schemes \cite{wong2022fluid, wong2023slow, hong2025coded}, where it has been revealed to support many users over the same channel without the need for precoding. Additionally, the versatility of FAS has been demonstrated across a range of emerging communication paradigms, including MEC \cite{zou2024fluid}, ISAC \cite{zou2024shifting}, and physical-layer security \cite{tang2023fluid}, further underscoring its value as a foundational technology for next-generation wireless networks.

Motivated by the concept of FAS, Ye {\em et al.}~\cite{ye2025joint} introduced the idea of position reconfigurability into RIS systems and proposed the concept of fluid RIS (FRIS). They applied FRIS to an ISAC system and demonstrated, through extensive simulations, that it significantly outperforms conventional RIS in terms of system performance. However, the implementation in \cite{ye2025joint} relies on physically movable elements to achieve position reconfigurability, which presents practical challenges. In particular, issues such as hardware complexity, slow mechanical response, and concerns over reliability and maintenance raise questions about its feasibility in real-world systems. To address these challenges, Salem {\em et al.}~\cite{salem2025first} revisited the FRIS concept by redefining the ``fluid'' element as a subregion of the RIS, wherein the reflection element can be dynamically switched among densely spaced ports without requiring physical movement. The results in \cite{salem2025first} illustrated substantial performance improvements over traditional RIS, highlighting the promise of FRIS in a more practical form. Building on this work, Xiao {\em et al.}~\cite{xiao2025fluid} proposed an even more implementation-friendly FRIS architecture, developed within the framework of conventional RIS-aided systems. In this approach, the entire RIS is treated as a single fluid element, and the system dynamically selects and activates reflecting elements at specific positions based on instantaneous channel conditions. Importantly, the design also considers the hardware limitations of discrete phase shifts typically found in practical RIS situations. This added realism further enhances the feasibility of the proposed scheme and strengthens its potential for practical deployment. Recently, the FAS concept has also been integrated with the simultaneous transmitting and reflecting surface (STAR)-RIS (STAR-RIS), leading to the term, known as fluid integrated reflecting and emitting surfaces (FIRES) introduced in \cite{Ghadi-2025fires}.

\vspace{-2mm}
\subsection{Motivations and Contributions}
Although existing FRIS schemes clearly show the potential of introducing position-reconfigurable DoF into traditional RIS-assisted systems for performance enhancement, accurate channel state information (CSI) is essential for selecting the optimal reflecting element positions. Recently, some channel estimation techniques originally developed for FAS, such as those in \cite{xu2024channel,wee2025channel,zhang2025successive}, exploit spatial correlation to reduce the dimensionality of the CSI estimation problem and can be adapted to the FRIS framework. Nevertheless, one major issue of the current FAS and FRIS research is that each point of radiation or selected port is treated as a point source of zero dimension with an omnidirectional pattern. This simplified modelling was useful in avoiding consideration of the intricate antenna structure and allowing us to focus on the diversity of a position-reconfigurable channel of FAS. That said, this clearly deviates from the reality and a selected port comes with a certain radiation pattern dictated by the antenna and element structure. It is therefore imperative to explore a new FRIS paradigm that incorporates a more realistic modelling of the element so that the performance of FRIS can be better understood.

Recently, there was a major development: Zhang {\em et al.}~\cite{zhang2025novel} developed an FAS prototype that is capable of generating $12$ distinct reconfigurable states---or {\em fluid patterns}---each representing a unique antenna configuration created through a specific pixel connection scheme. A reconfigurable state can be viewed as a functional equivalent of an antenna positioned at a different location. Notably, a closer examination reveals that, analogous to spatial diversity achieved through position reconfigurability, spatial diversity also manifests across the different channels induced by these fluid patterns. In this sense, within FAS systems, ``reconfigurable position'' and ``reconfigurable pattern'' are, to some extent, functionally interchangeable in terms of signal control. However, considering pattern instead of position provides a route to connect to the actual antenna structure for end-to-end performance analysis.

Motivated by this insight, pattern reconfigurability emerges as an alternative to position reconfigurability for overcoming the limitations of position-reconfigurable FRIS systems. As a result, this paper introduces a novel pattern-reconfigurable FRIS framework. Specifically, we first explore the potential of pattern reconfigurability for signal manipulation and provide an in-depth analysis of its underlying operating principles. To further demonstrate the practical value of the proposed framework, we apply it to a multiuser communication scenario and present extensive simulation results to systematically validate its effectiveness and performance gains. The main contributions of this paper are summarized as follows:
\begin{itemize}
\item {\em A First Look at the Performance Enhancement Potential of Pattern-Reconfigurable FRIS}---We begin with a preliminary study in a point-to-point RIS-assisted system, comparing the received signal power achieved by pattern-reconfigurable FRIS, position-reconfigurable FRIS, and traditional RIS with passive beamforming. The obtained findings indicate that both pattern reconfigurability and position reconfigurability can independently manipulate each multi-path component, whereas traditional passive beamforming can only apply a unified control over all multi-path signals, lacking fine-grained adjustment capabilities. This highlights the high flexibility of pattern and position reconfigurability in enhancing signal quality. Furthermore, we observe that, compared to position reconfigurability, pattern reconfigurability is more effective in enabling the received signal power to approach its theoretical upper bound. This suggests that pattern-reconfigurable FRIS could serve as a practical and efficient alternative to position-reconfigurable FRIS.
\item {\em Model and Problem Formulation}---To further exploit the capability of the pattern-reconfigurable FRIS, we incorporate it into a multiuser communication scenario. Considering that there is no explicit analytical mapping between the azimuth and elevation angles and the radiation pattern gain, optimizing the radiation pattern of fluid elements becomes a significant challenge. To solve this issue, the spherical harmonics orthogonal decomposition (SHOD) method is used to accurately model the radiation patterns of fluid elements, making pattern design more tractable. Subsequently, we formulate an optimization problem to maximize the weighted sum rate by jointly designing active beamforming and spherical harmonics coefficients, subject to the power budget at the base station (BS) and the energy constraints of the radiation patterns. 
\item {\em Effective Algorithm with Guaranteed Convergence}---The formulated problem is non-convex due to the strong coupling between variables and the inherent non-convexity of both the objective function and the constraints. To address this problem, an iterative algorithm based on minimum mean-square error (MMSE) method and Riemannian conjugate gradient (RCG) method is proposed to alternatively solve the active beamforming and spherical harmonics coefficients. The convergence of the proposed algorithm is guaranteed and validated through simulation results. 
\item {\em Performance Improvement}---To quantify the performance gain achievable by the pattern-reconfigurable FRIS, extensive simulations are conducted. The obtained results indicate that compared with traditional RIS employing 3GPP 38.901 radiation patterns and isotropic patterns, the proposed pattern-reconfigurable FRIS achieves average performance gains of $161.5\%$ and $176.2\%$, respectively. Moreover, it can reduce the required number of antennas and elements by over $300\%$ while maintaining the same performance level, demonstrating exceptional hardware efficiency. These results highlight the unparalleled potential of the FRIS in enhancing system performance.
\end{itemize}


\textit{Notations}---$\mathbb{Z}$ denotes the integer set. The operations $(\cdot)^T$, $(\cdot)^*$  and $(\cdot)^H$ denote the transpose, conjugate and conjugate transpose, correspondingly.  $\operatorname{Diag}(\mathbf{a})$ denotes a diagonal matrix whose diagonal elements are composed of the vector $\mathbf{a}$.  $\operatorname{Blkdiag}\{\mathbf{a}_1, \mathbf{a}_2, \dots, \mathbf{a}_M\}$ represents a block diagonal matrix. $\operatorname{Ddiag}(\mathbf{A})=\operatorname{Diag}(\operatorname{diag}(\mathbf{A}))$. Additionally, $\mathbf{I}_{M}$ and $\mathbf{1}_{M\times 1}$ represent an $M\times M$ identity matrix and $M\times 1$ vector of ones, respectively. Also, the symbols $|\cdot|$  and $\|\cdot\|_2$ are indicative of the complex modulus and spectral norm, respectively. Moreover, operator $\circ$ and $\otimes$ denote the Hadamard product and Kronecker product. Finally, $\operatorname{R}(\cdot)$ denotes the real part operator.

\vspace{-2mm}
\section{A Preliminary Study on Pattern-Reconfigurable FRIS}\label{sec:S2} 
In this section, we provide a preliminary investigation into the performance gains offered by pattern-reconfigurable FRIS, through a comparative analysis of the received signal power for pattern-reconfigurable FRIS, position-reconfigurable FRIS and traditional RIS with passive beamforming. This analysis is based on the principle that all three schemes fundamentally operate by manipulating signal phase to enhance the channel quality between the transmitter and the receiver. To enable an intuitive and effective comparison, a simple point-to-point communication scenario is considered.

Specifically, the system consists of a BS transmitter, a user receiver, and an RIS equipped with $M$ elements. The RIS is assumed to operate in one of the following modes to modulate the incident signals: pattern-reconfigurable-only (pattern-reconfigurable FRIS), position-reconfigurable-only (position-reconfigurable FRIS), or passive beamforming-only (i.e., traditional RIS). It is assumed that the signal transmission from the transmitter to the receiver relies entirely on the assistance of the RIS, with no direct link between them.

The channels between RIS and transmitter/receiver can be modeled as \cite{ying2024reconfigurable}
\begin{equation}
\mathbf{h}_\mathrm{br}=\sum_{l=1}^{L}\mathbf{h}_{\mathrm{br}, l},~\mbox{and}~
\mathbf{h}_\mathrm{ru}=\sum_{z=1}^{Z}\mathbf{h}_{\mathrm{ru}, z},
\end{equation}
with
\begin{equation}
\left\{\begin{aligned}
\mathbf{h}_{\mathrm{br}, l}&=[g_{\mathrm{br}, l}f_1(\theta_{\mathrm{r}, l}, \phi_{\mathrm{r},l})e^{j\frac{2\pi}{\lambda}\mathbf{k}^T_{\mathrm{r}, l}\mathbf{p}_1}, \dots, \\
&\quad\quad\quad\quad g_{\mathrm{br}, l} f_M(\theta_{\mathrm{r}, l}, \phi_{\mathrm{r},l}) e^{j\frac{2\pi}{\lambda}\mathbf{k}^T_{\mathrm{r}, l}\mathbf{p}_M}]^T\in\mathbb{C}^{M\times 1},\\
\mathbf{h}_{\mathrm{ru}, z}&=[g_{\mathrm{ru}, z}f_1(\theta_{\mathrm{t}, z}, \phi_{\mathrm{t},z})e^{j\frac{2\pi}{\lambda}\mathbf{k}^T_{\mathrm{t}, z}\mathbf{p}_1}, \dots, \\
&\quad\quad\quad\quad g_{\mathrm{ru}, z} f_M(\theta_{\mathrm{t}, z}, \phi_{\mathrm{t},z})e^{j\frac{2\pi}{\lambda}\mathbf{k}^T_{\mathrm{t}, z}\mathbf{p}_M}]^T\in\mathbb{C}^{M\times 1},
\end{aligned}\right.
\end{equation}
where $L$ and $Z$ are the number of propagation paths from the BS transmitter to the RIS, and from the RIS to the receiver, respectively. Also, $g_{\mathrm{br}, l}$ and $g_{\mathrm{ru}, z}$ are the complex channel gain of the $l$-th path from the transmitter to RIS and the $z$-th path from the RIS to receiver, respectively. Additionally, $\mathbf{p}_m$, for $m\in\mathcal{M}\triangleq\{1, \dots, M\}$, is the $m$-th element's position, $f_m(\theta_{\mathrm{t},z}, \phi_{\mathrm{t},z})$ and $f_m(\theta_{\mathrm{r},l}, \phi_{\mathrm{r},l})$ denote the transmitting  pattern gain and  receiving pattern gain of the $m$-th element. Also, $\mathbf{k}_{\mathrm{r}, l}=[\sin\theta_{\mathrm{r},l}\sin\phi_{\mathrm{r},l}, \sin\theta_{\mathrm{r},l}\cos\phi_{\mathrm{r},l}, \cos\theta_{\mathrm{r},l}]^T$ and $\mathbf{k}_{\mathrm{t}, z}=[\sin\theta_{\mathrm{t}, z}\sin\phi_{\mathrm{t}, z}, \sin\theta_{\mathrm{t}, z}\cos\phi_{\mathrm{t}, z}, \cos\theta_{\mathrm{t}, z}]^T$ are the normalized wave vectors, and $\lambda$ is the carrier wavelength. The following four cases are utilized to unveil the performance enhancement potential of pattern-reconfigurable FRIS.

{\em Case 1: Single-element Single-path Scenario}---Consider in this case $L=Z=M=1$. Then the received signal can be expressed as
\begin{align}
y&=g^*_{\mathrm{ru}}f^*(\theta_{\mathrm{t}}, \phi_{\mathrm{t}})e^{-j\frac{2\pi}{\lambda}\mathbf{k}_{\mathrm{t}}^T\mathbf{p}}\vartheta g_{\mathrm{br}}
f(\theta_{\mathrm{r}}, \phi_{\mathrm{r}})e^{j\frac{2\pi}{\lambda}\mathbf{k}_{\mathrm{r}}^T\mathbf{p}}s +n\notag\\
&=gf_\mathrm{r, t}(\theta, \phi)\vartheta e^{j\frac{2\pi}{\lambda}(\mathbf{k}_{\mathrm{r}}- \mathbf{k}_{\mathrm{t}})^T\mathbf{p}}s +n,
\end{align}
where $g\triangleq g^*_{\mathrm{ru}}g_{\mathrm{br}}$, $f_\mathrm{r, t}(\theta, \phi)\triangleq f^*(\theta_{\mathrm{t}}, \phi_{\mathrm{t}})f(\theta_{\mathrm{r}}, \phi_{\mathrm{r}})$, $s$ is the transmitting symbol with $\mathbb{E}[|s|^2]=1$, $n\sim{\cal CN}(0,\sigma^2)$ is the additive white Gaussian noise (AWGN) with power $\sigma^2$, $\vartheta$ denotes the reflection coefficient. Note that we temporarily neglect the index $l$, $z$ and $m$ for conciseness.  Thus, the received power can be given by
$
	P_\mathrm{r}=\left|gf_\mathrm{r, t}(\theta, \phi)\vartheta e^{j\frac{2\pi}{\lambda}(\mathbf{k}_{\mathrm{r}}- \mathbf{k}_{\mathrm{t}})^T\mathbf{p}}\right|^2+\sigma^2.
$ 
It is observed that irrespective of variations in the radiation pattern, element positions, or reflection coefficients, the received power remains invariant in this scenario. This invariance indicates that the potential of pattern-reconfigurable FRIS cannot be effectively validated under such a setup.

{\em Case 2: Single-element Multi-path Scenario}---In this case, $M=1$ and we have the received signal expressed as
\begin{align}
y&=\sqrt{\frac{1}{LZ}}\sum_{l=1}^{L}\sum_{z=1}^{Z}g^*_{\mathrm{ru}, z}f^*(\theta_{\mathrm{t}, z}, \phi_{\mathrm{t}, z})e^{-j\frac{2\pi}{\lambda}\mathbf{k}_{\mathrm{t}, z}^T\mathbf{p}}\vartheta g_{\mathrm{br}, l}\notag\\
&\quad\quad\quad\quad\quad\quad\quad \times f(\theta_{\mathrm{r}, l}, \phi_{\mathrm{r}, l})e^{j\frac{2\pi}{\lambda}\mathbf{k}_{\mathrm{r}, l}^T\mathbf{p}}s +n\notag\\
&=\sqrt{\frac{1}{LZ}}\sum_{l=1}^{L}\sum_{z=1}^{Z}g_{l, z}f^{l, z}_\mathrm{r, t}(\theta, \phi)\vartheta e^{j\frac{2\pi}{\lambda}(\mathbf{k}_{\mathrm{r}, l}- \mathbf{k}_{\mathrm{t}, z})^T\mathbf{p}}s +n.
\end{align}
 Thus, the power of the received signal is given by
\begin{align}\label{eq_receive_power_2}
P_\mathrm{r}=\left|\sqrt{\frac{1}{LZ}}\sum_{l=1}^{L}\sum_{z=1}^{Z}g_{l, z}f^{l, z}_\mathrm{r, t}(\theta, \phi)\vartheta e^{j\frac{2\pi}{\lambda}(\mathbf{k}_{\mathrm{r}, l}- \mathbf{k}_{\mathrm{t}, z})^T\mathbf{p}}\right|^2+\sigma^2.
\end{align}

\begin{corollary}\label{coro1}
With the assistance of passive beamforming, the optimal received power can be expressed as
\begin{align}
P^\mathrm{opt}_\mathrm{r, pb} = \frac{1}{LZ} \left|\sum_{l=1}^{L}\sum_{z=1}^{Z} g_{l, z} f^{l, z}_\mathrm{r, t}(\theta, \phi) e^{j\frac{2\pi}{\lambda}(\mathbf{k}_{\mathrm{r}, l}- \mathbf{k}_{\mathrm{t}, z})^T\mathbf{p}}\right|^2+\sigma^2.
\end{align}
The corresponding optimal reflection coefficient should satisfy
\begin{multline}
\angle \vartheta = -\angle \left( \sum_{l=1}^{L} \sum_{z=1}^{Z} g_{l, z} f^{l, z}_\mathrm{r, t}(\theta, \phi)e^{j\frac{2\pi}{\lambda}(\mathbf{k}_{\mathrm{r}, l}- \mathbf{k}_{\mathrm{t}, z})^T\mathbf{p}} \right)\\
+2k_1\pi, ~k_1\in\mathbb{Z}.
\end{multline}
\end{corollary}

\begin{corollary}\label{coro2}
When the reflection coefficient and radiation pattern are fixed, in order to maximize the received power, the element position should satisfy
\begin{multline}\label{eq_position_opt_1}
\frac{2\pi}{\lambda}(\mathbf{k}_{\mathrm{r}, l}- \mathbf{k}_{\mathrm{t}, z})^T\mathbf{p}=-\angle g_{l, z}f^{l, z}_\mathrm{r, t}(\theta, \phi)\vartheta\\
+2k_2\pi,~\forall l\in\mathcal{L}, z\in\mathcal{Z}, k_2\in\mathbb{Z},
\end{multline}  
which gives the received power 
\begin{equation}\label{eq_position_upper}
P^\mathrm{opt}_\mathrm{r, pp} = \frac{1}{LZ} \left( \sum_{l=1}^{L} \sum_{z=1}^{Z} \left|g_{l, z} f^{l, z}_\mathrm{r, t}(\theta, \phi)\right| \right)^2+\sigma^2.
\end{equation}
\end{corollary}

Actually, it is challenging to find a feasible position $\mathbf{p}$ that satisfies all the conditions specified in \eqref{eq_position_opt_1}. Therefore, the result in \eqref{eq_position_upper} serves only as a theoretical upper bound of the received power in the position reconfigurability scheme.

\begin{corollary}\label{coro3}
When the reflection coefficient and position are fixed, and the radiation pattern is designed accordingly, the optimal received power can be derived as
\begin{align}
P^\mathrm{opt}_\mathrm{r, pr} = \frac{1}{LZ} \left( \sum_{l=1}^{L} \sum_{z=1}^{Z} \left|g_{l, z} f^{l, z}_\mathrm{r, t}(\theta, \phi)\right| \right)^2+\sigma^2,
\end{align}
which is a theoretic upper bound of the received power in \eqref{eq_receive_power_2}. The corresponding optimal pattern should satisfy
\begin{equation}
\angle f^{l, z}_\mathrm{r, t}(\theta, \phi) = -\angle \left(\vartheta g_{l, z}  e^{j\frac{2\pi}{\lambda}(\mathbf{k}_{\mathrm{r}, l}- \mathbf{k}_{\mathrm{t}, z})^T\mathbf{p}}\right),~ \forall l\in\mathcal{L}, z\in\mathcal{Z}.
\end{equation}
\end{corollary}

{\em Case 3: Multi-element Single-path Scenario}---Here we have $L=Z=1$. Thus, the received signal is found as
\begin{equation}
y=\sum_{m=1}^{M}\vartheta_mg_{m}f^{m}_\mathrm{r, t}(\theta, \phi) e^{j\frac{2\pi}{\lambda}(\mathbf{k}_{\mathrm{r}}- \mathbf{k}_{\mathrm{t}})^T\mathbf{p}_m}s +n.
\end{equation}
In this case, the reflection coefficients, element positions, and radiation patterns can be independently designed to enable the received power to approach the theoretical upper bound
\begin{align}\label{eq_receive_power_3}
P^\mathrm{opt}_\mathrm{r}=\left(\sum_{m=1}^{M}\left|g_{m}f^{m}_\mathrm{r, t}(\theta, \phi)\right|\right)^2+\sigma^2.
\end{align}

{\em Case 4: Multi-element Multi-path Scenario}---The represents the most general case, with the received signal given by
\begin{multline}
y=\sqrt{\frac{1}{LZ}}\sum_{m=1}^{M}\vartheta_m\sum_{l=1}^{L}\sum_{z=1}^{Z}g_{m, l, z}f^{m, l, z}_\mathrm{r, t}(\theta, \phi)\\
\times e^{j\frac{2\pi}{\lambda}(\mathbf{k}_{\mathrm{r}, m, l}- \mathbf{k}_{\mathrm{t}, m, z})^T\mathbf{p}_m}s +n.
\end{multline}
The received power can then be expressed as
\begin{multline}
P_\mathrm{r}=\left|\sqrt{\frac{1}{LZ}}\sum_{m=1}^{M}\vartheta_m\sum_{l=1}^{L}\sum_{z=1}^{Z}g_{m, l, z}f^{m, l, z}_\mathrm{r, t}(\theta, \phi)\right.\\
\left.\times e^{j\frac{2\pi}{\lambda}(\mathbf{k}_{\mathrm{r}, m, l}- \mathbf{k}_{\mathrm{t}, m, z})^T\mathbf{p}_m}\right|^2+\sigma^2.
\end{multline}

\begin{corollary}\label{coro4}
With the assistance of passive beamforming, the optimal received power can be expressed as
\begin{multline}
P^\mathrm{opt}_\mathrm{r, pb} = \frac{1}{LZ} \left(\sum_{m=1}^{M}\left|\sum_{l=1}^{L}\sum_{z=1}^{Z}g_{m, l, z}f^{m, l, z}_\mathrm{r, t}(\theta, \phi)\right.\right.\\
\left.\left.\times e^{j\frac{2\pi}{\lambda}(\mathbf{k}_{\mathrm{r}, m, l}- \mathbf{k}_{\mathrm{t}, m, z})^T\mathbf{p}_m}\right|\right)^2+\sigma^2,
\end{multline}
with
\begin{multline}
\angle \vartheta_m = -\angle \left(\sum_{l=1}^{L}\sum_{z=1}^{Z}f^{m, l, z}_\mathrm{r, t}(\theta, \phi)e^{j\frac{2\pi}{\lambda}(\mathbf{k}_{\mathrm{r}, m, l}- \mathbf{k}_{\mathrm{t}, m, z})^T\mathbf{p}_m}\right.\\
\left.\times g_{m, l, z}\right)+2k_3\pi,~\forall m\in\mathcal{M}, l\in\mathcal{L}, z\in\mathcal{Z}, k_3\in\mathbb{Z}.
\end{multline}
\end{corollary}

\begin{corollary}\label{coro5}
When the reflection coefficient and radiation pattern are fixed, in order to maximize the received power, the element position should satisfy
\begin{multline}
\frac{2\pi}{\lambda}(\mathbf{k}_{\mathrm{r}, m, l}- \mathbf{k}_{\mathrm{t}, m, z})^T\mathbf{p}_m\\
=-\angle g_{m, l, z}f^{m, l, z}_\mathrm{r, t}(\theta, \phi)\vartheta_m+2k_4\pi,\\
\forall m\in\mathcal{M}, l\in\mathcal{L}, z\in\mathcal{Z}, k_4\in\mathbb{Z}.
\end{multline}  
In this case, the received power is obtained as
\begin{equation}
P^\mathrm{opt}_\mathrm{r, pp} = \frac{1}{LZ} \left(\sum_{m=1}^{M} \sum_{l=1}^{L} \sum_{z=1}^{Z} \left|g_{m, l, z} f^{m, l, z}_\mathrm{r, t}(\theta, \phi)\right| \right)^2+\sigma^2.
\end{equation}
\end{corollary}

Similar to the position reconfigurability scheme, it is challenging to find a feasible position $\mathbf{p}_m$, $m\in\mathcal{M}$ that allows the received power to reach its theoretical upper bound.

\begin{corollary}\label{coro6}
Similarly, the optimal received power considering the design of radiation pattern can be derived as
\begin{equation}
P^\mathrm{opt}_\mathrm{r, pr} =\frac{1}{LZ} \left( \sum_{m=1}^{M}\sum_{l=1}^{L} \sum_{z=1}^{Z} \left|g_{m, l, z} f^{m, l, z}_\mathrm{r, t}(\theta, \phi)\right| \right)^2+\sigma^2,
\end{equation}
which is also the theoretic upper bound of the received power. Also, the corresponding optimal pattern must satisfy
\begin{multline}
\angle f^{m, l, z}_\mathrm{r, t}(\theta, \phi)\\
= -\angle \left(\vartheta_m g_{m, l, z} e^{j\frac{2\pi}{\lambda}(\mathbf{k}_{\mathrm{r}, m, l}- \mathbf{k}_{\mathrm{t}, m, z})^T\mathbf{p}_m}\right)+2k_5\pi,\\
\forall m\in\mathcal{M}, l\in\mathcal{L}, z\in\mathcal{Z}, k_5\in\mathbb{Z}.
\end{multline}
\end{corollary}

\begin{remark}
The above analysis shows that in the single-path scenario, the received power can attain its theoretical upper bound when the reflection coefficients, element positions, and radiation patterns are independently and properly designed (as shown in Case 1 and Case 3). This observation confirms the functional equivalence among passive beamforming, position reconfigurability, and pattern reconfigurability in RIS-aided systems. In multi-path scenarios, nevertheless, passive beamforming appears to be limited to adjusting only the overall phase of the cascaded channel, which restricts its ability to enhance the received power. In contrast, position and pattern reconfigurabilities enable fine-grained phase control over individual multi-path components, offering greater DoFs for signal manipulation and thereby achieving improved theoretical performance. Furthermore, Corollaries \ref{coro2} and \ref{coro5} suggest that achieving the theoretical upper bound is challenging under the position-reconfigurable scheme due to the difficulty of finding the feasible positions to meet all the given constraints, while Corollaries \ref{coro3} and \ref{coro6} indicate that the pattern-reconfigurable scheme can readily achieve this bound. 

These findings highlight that pattern reconfigurability is sufficient to effectively replace both passive beamforming and position reconfigurability in RIS systems, demonstrating pattern-reconfigurable FRIS's remarkable potential in enhancing communication performance compared to the other two.\footnote{Pattern reconfigurability is also a more relatable concept in antennas since it connects more closely to the intrinsic properties of antenna.} Notably, combining pattern reconfigurability with passive beamforming offers no additional performance gains because pattern reconfigurability alone provides ample design flexibility for optimal signal alignment. Therefore, in the pattern-reconfigurable FRIS  considered in this paper, our design focus is solely placed on radiation pattern optimization.\footnote{In practical applications, achieving optimal performance through pattern reconfigurability remains challenging, primarily due to hardware limitations that restrict the number of available reconfigurable patterns. Therefore, it is necessary to jointly design the passive beamforming and radiation pattern to further enhance the system performance, which represents a promising future direction for the research on pattern reconfigurability-enabled FRIS.} It is worth noting that in the above example, we only explored the capability of radiation patterns to manipulate the phase of the transmitted signals. However, the potential of radiation patterns goes far beyond that. In fact, radiation patterns also possess the ability to focus signal energy, which means they can not only control the signal phase but also effectively regulate the amplitude distribution of the transmitted signals.
\end{remark}


\section{Multiuser System Model}\label{sec:S3}
In this sections, we will integrate the proposed pattern reconfigurability-enabled FRIS into wireless networks to validate its  potential in enhancing communication performance of wireless networks. Specifically, Fig.~\ref{fig:scenario} shows the considered communication system, which consists of a multi-antenna BS, where each antenna is fixed and assumed to have an isotropic radiation pattern (i.e., the gain is $1$ in all directions), an FRIS with $M$ pattern-reconfigurable elements (fluid elements) and $K$ single-fixed-antenna user equipment (UEs), each antenna with an isotropic radiation pattern. To investigate the potential of the pattern-reconfigurability-enabled FRIS, it is assumed that the direct links between the BS and the UEs are completely blocked, and all signals are transmitted to the UEs exclusively through the assistance of the FRIS.

\begin{figure}[]
\centering
\includegraphics[width=.65\columnwidth]{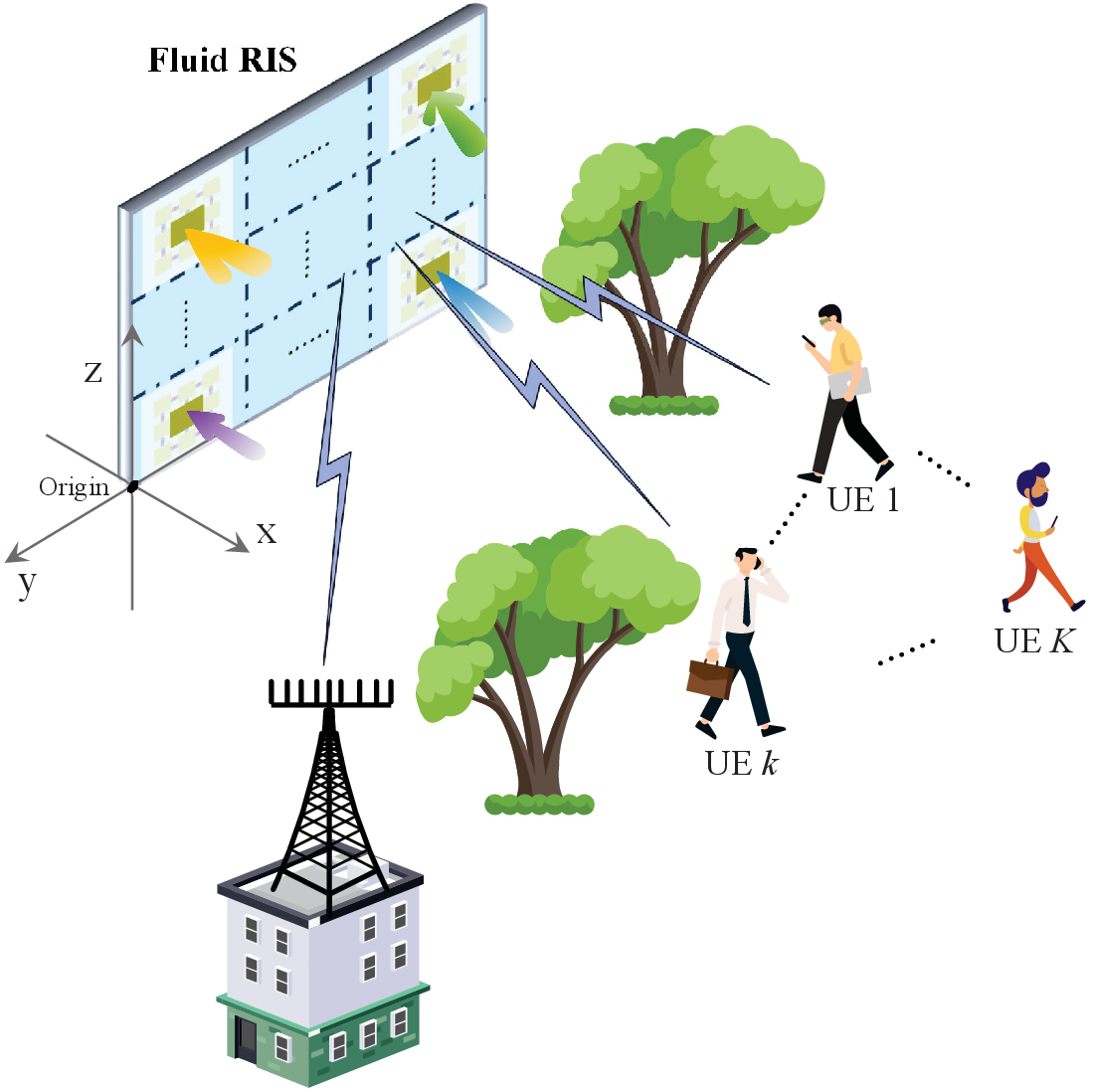}\\
\vspace{-2mm}
\caption{A downlink system with a pattern-reconfigurable FRIS.}\label{fig:scenario}
\vspace{-3mm}
\end{figure}

\subsection{Channel Model}
A multi-path channel model is used to represent the links between the FRIS and the BS/UEs. Specifically, the channel between the BS and the FRIS is written as
\begin{equation}
\mathbf{H}_{\mathrm{BR}}=\gamma_1\left(\mathbf{H}_{\mathrm{BR}}^\mathrm{LoS}+\mathbf{H}_{\mathrm{BR}}^\mathrm{NLoS}\right),
\end{equation}
where
\begin{equation}
\left\{\begin{aligned}
\mathbf{H}_{\mathrm{BR}}^\mathrm{LoS}&=\rho(d^\mathrm{LoS}_\mathrm{BR})g^\mathrm{LoS}_{\mathrm{br}}\mathbf{P}_\mathrm{RIS}\left(\theta_\mathrm{r}^\mathrm{LoS}, \phi_\mathrm{r}^\mathrm{LoS}\right)\mathbf{a}_\mathrm{R}\left(\theta_\mathrm{r}^\mathrm{LoS}, \phi_\mathrm{r}^\mathrm{LoS}\right)\\
&\times\mathbf{a}^H_\mathrm{B}\left(\theta_\mathrm{t, B}^\mathrm{LoS}, \phi_\mathrm{t, B}^\mathrm{LoS}\right),\\
\mathbf{H}_{\mathrm{BR}}^\mathrm{NLoS}&=\sum_{l=1}^{L}\rho(d^l_\mathrm{BR})g^l_{\mathrm{br}}\mathbf{P}_\mathrm{RIS}\left(\theta_\mathrm{r}^{l}, \phi_\mathrm{r}^{l}\right)\mathbf{a}_\mathrm{R}\left(\theta_\mathrm{r}^{l}, \phi_\mathrm{r}^{l}\right)\\
&\times\mathbf{a}^H_\mathrm{B}\left(\theta_\mathrm{t, B}^{l}, \phi_\mathrm{t, B}^{l}\right),
\end{aligned}\right.
\end{equation}
in which $\gamma_1=\sqrt{\frac{1}{L+1}}$ with $L$ denoting the number of multi-path components from the BS to the FRIS, $\rho(d_0)=\sqrt{\frac{\rho_0}{d_0^{\alpha}}}$ represents the large scale path-loss, which is a function with respect to (w.r.t.) the distance, $d_0$, where $\rho_0$ and $ \alpha$ denote the reference power gain at a distance of $1$ m and the path-loss exponent, respectively. Moreover, $\mathbf{P}_\mathrm{RIS}\left(\theta_\mathrm{r}^\zeta, \phi_\mathrm{r}^\zeta\right)\in \mathbb{C}^{M\times M}, ~\zeta\in\{\mathrm{LoS}, \{l\}_{l=1}^L\}$, denotes a diagonal matrix where $[\mathbf{P}_\mathrm{RIS}\left(\theta_\mathrm{r}^\zeta, \phi_\mathrm{r}^\zeta\right)]_{m, m}=p_m\left(\theta_\mathrm{r}^\zeta, \phi_\mathrm{r}^\zeta\right), ~m\in\mathcal{M}\triangleq\{1, \dots, M\},$ is the radiation pattern gain of the $m$-th fluid element in the direction $\left(\theta_\mathrm{r}^\zeta, \phi_\mathrm{r}^\zeta\right)$ with $\theta_\mathrm{r}^\zeta$ and $\phi_\mathrm{r}^\zeta$ representing the elevation and  azimuth angles of arrival (AoA) at the FRIS. Furthermore, $\theta_\mathrm{t, B}^{\zeta}$ and $\phi_\mathrm{t, B}^{\zeta}$ represent the elevation and  azimuth angles of departure (AoD) at the BS, $d_\mathrm{BR}^\zeta$ denotes the multi-path length, $g_\mathrm{br}^\zeta\sim\mathcal{CN}(0,1)$ is the complex channel gain. Also, 
\begin{equation}
\left\{\begin{aligned}\mathbf{a}_\mathrm{R}\left(\theta_\mathrm{r}^\zeta, \phi_\mathrm{r}^\zeta\right)&=e^{-j2\pi\frac{\mathbf{t}^y_\mathrm{R}d\sin\theta_\mathrm{r}^\zeta\sin\phi_\mathrm{r}^\zeta}{\lambda}}e^{-j2\pi\frac{\mathbf{t}^z_\mathrm{R}d\cos\theta_\mathrm{r}^\zeta}{\lambda}},\\
\mathbf{a}_\mathrm{B}\left(\theta_\mathrm{t, B}^\zeta, \phi_\mathrm{t, B}^\zeta\right)&=e^{-j2\pi\frac{\mathbf{t}_\mathrm{B}d\sin\theta_\mathrm{t, B}^\zeta\cos\phi_\mathrm{t, B}^\zeta}{\lambda}},
\end{aligned}\right.
\end{equation}
where $\mathbf{t}^y_\mathrm{R}=[0, 1, \dots, M_y-1]^T$, $\mathbf{t}^z_\mathrm{R}=[0, 1, \dots, M_z-1]^ T$, $\mathbf{t}_\mathrm{B}=[0, 1, \dots, N_\mathrm{t}-1]^T$, and $d=\frac{\lambda}{2}$ is assumed to be the distance between adjacent elements.  

Similarly, the channel between the FRIS and the $k$-th UE, i.e., $\mathbf{h}_{\mathrm{r}k}, ~k\in\mathcal{K}\triangleq\{1, \dots, K\}$, can be expressed as
\begin{equation}
\mathbf{h}_{\mathrm{r}k}=\gamma_2\left(\mathbf{h}_{\mathrm{r}k}^\mathrm{LoS}+\mathbf{h}_{\mathrm{r}k}^\mathrm{NLoS}\right),
\end{equation}
where
\begin{equation}
\left\{\begin{aligned}
\mathbf{h}_{\mathrm{r}k}^\mathrm{LoS}&=\rho(d^{\mathrm{LoS}}_{\mathrm{r}k})g_{\mathrm{r}k}^{\mathrm{LoS}}\mathbf{P}_\mathrm{RIS}\left(\theta_{\mathrm{t}, k}^{\mathrm{LoS}}, \phi_{\mathrm{t}, k}^{\mathrm{LoS}}\right)
\mathbf{a}_\mathrm{R}\left(\theta_{\mathrm{t}, k}^{\mathrm{LoS}}, \phi_{\mathrm{t}, k}^{\mathrm{LoS}}\right),\\
\mathbf{h}_{\mathrm{r}k}^\mathrm{NLoS}&=\sum_{z=1}^Z \rho(d^z_{\mathrm{r}k}) g_{\mathrm{r}k}^z~\mathbf{P}_\mathrm{RIS}\left(\theta_{\mathrm{t}, k}^{z}, \phi_{\mathrm{t}, k}^{z}\right)\mathbf{a}_\mathrm{R}\left(\theta_{\mathrm{t}, k}^{z}, \phi_{\mathrm{t}, k}^{z}\right),
\end{aligned}\right.
\end{equation}
in which $\gamma_2=\sqrt{\frac{1}{Z+1}}$ with $Z$ denoting the number of multi-path components from the FRIS to UE $k$,  $\theta_{\mathrm{t}, k}^{\widehat{\zeta}}$ and $\phi_{\mathrm{t}, k}^{\widehat{\zeta}}$, ${\widehat{\zeta}}\in\{\mathrm{LoS}, \{z\}_{z=1}^Z\}$, are the elevation and azimuth AoDs from the FRIS to the $k$-th UE, respectively, while $\theta_{\mathrm{r}, k}^{\widehat{\zeta}}$ and $\phi_{\mathrm{r}, k}^{\widehat{\zeta}}$ are the corresponding AoAs at UE $k$, $d_\mathrm{BR}^{\widehat{\zeta}}$ denotes the multi-path length, and $g_\mathrm{br}^{\widehat{\zeta}}\sim\mathcal{CN}(0,1)$ is the complex channel gain. 

With this model, we have the $k$-th UE's received signal as
\begin{equation}
y_k=\mathbf{h}_{\mathrm{r}k}^H\mathbf{H}_\mathrm{BR}\mathbf{w}_ks_k+\sum_{j\neq k}^{K}\mathbf{h}_{\mathrm{r}k}^H\mathbf{H}_\mathrm{BR}\mathbf{w}_js_j+n_k,
\end{equation}
where $s_k$ is the transmitted signal at the BS for the $k$-th UE, with $\mathbb{E}[|s_k|^2]=1$, $\mathbf{w}_k$ denotes the active beamforming vector of UE $k$ at the BS, and $n_k\sim(0, \sigma_k^2)$ is the AWGN. Thus, the achievable rate of the $k$-th UE is given by
\begin{equation}
R_k=\log_2\left(1+\frac{\left|\mathbf{h}_{\mathrm{r}k}^H\mathbf{H}_\mathrm{BR}\mathbf{w}_k\right|^2}{\sum_{j\neq k}^{K}\left|\mathbf{h}_{\mathrm{r}k}^H\mathbf{H}_\mathrm{BR}\mathbf{w}_j\right|^2+\sigma_k^2}\right).
\end{equation}

\begin{figure*}[!t]
\centering
\includegraphics[width=.7\linewidth]{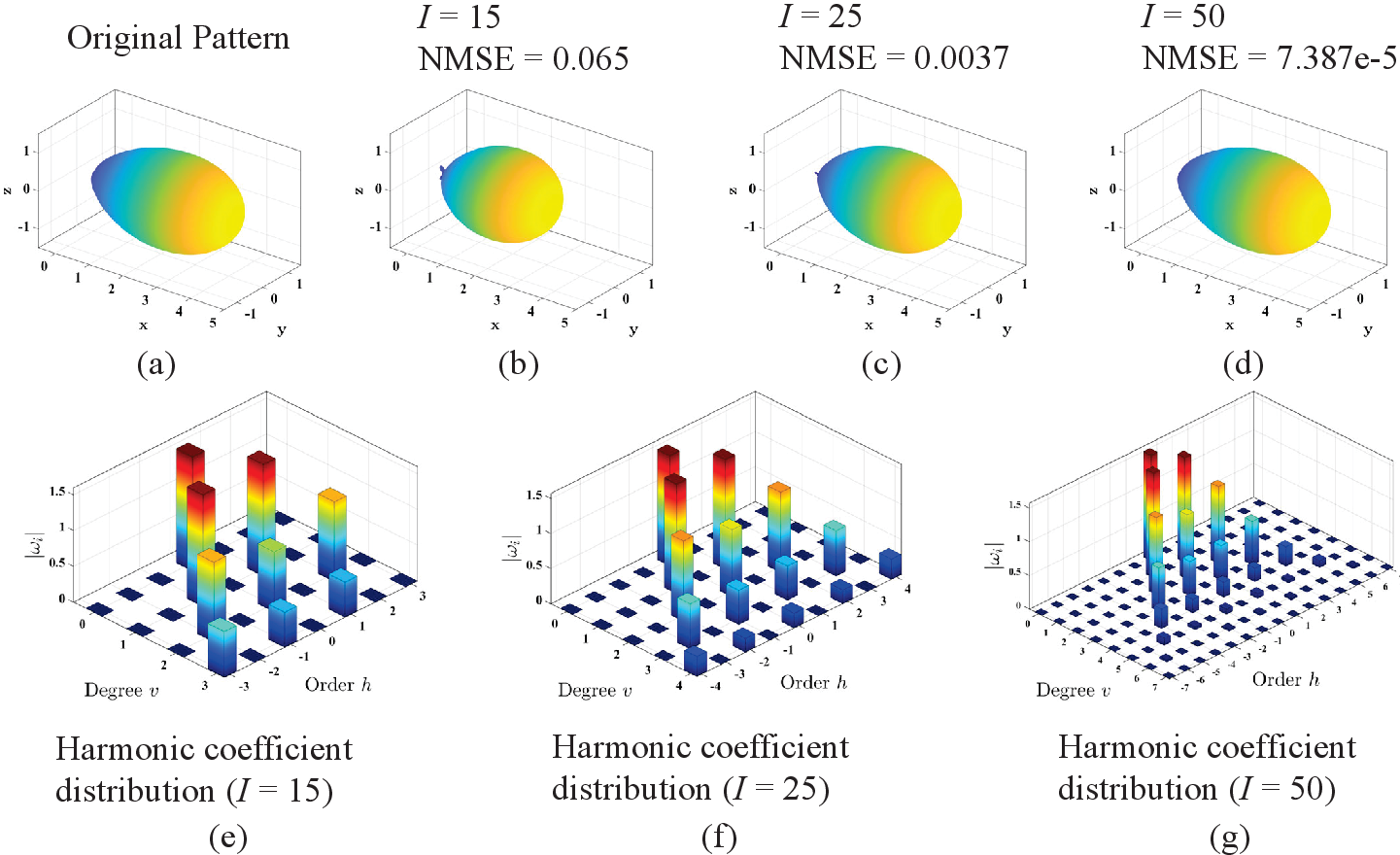}\\
\vspace{-2mm}
\caption{The illustration of a numerical example using the SHOD method for the decomposition of 3GPP 38.901 pattern considering different truncation length: (a) the original pattern of 3GPP 38.901; the SHOD approximation pattern (b) with $I=15$; (c) with $I=25$; (d) with $I=50$; the harmonic coefficient distribution (e) with $I=15$; (f) with $I=25$; (g) with $I=50$.}\label{fig:SHOD_method_verification}
\vspace{-4mm}
\end{figure*}

\subsection{SHOD for Radiation Pattern Reconfigurability}
The fact that there is no explicit analytical mapping between the azimuth and elevation angles and the radiation pattern gain in our model, presents a significant challenge in optimizing the three-dimensional (3D) radiation pattern of fluid elements. To overcome this, the SHOD method in \cite{costa2010unified, xu20163} is employed. Specifically, for any radiation pattern of the $m$-th element, it can be decomposed as 
\begin{equation}\label{eq_SHOD}
p_m\left(\theta, \phi\right)=\sum_{v=1}^{+\infty}\sum_{h=-v}^{v}\omega_{m, v, h}Y^h_{v}\left(\theta, \phi\right),
\end{equation}
where $\omega_{m, v, h}$ denotes the harmonic coefficient, and
\begin{equation}\label{eq_spherical_harmonic_basis}
Y^h_v(\theta, \phi)=\begin{cases}
(-1)^h N_v^hP^h_v(\cos\theta)e^{jh\phi}, &h\geq 0,\\
(-1)^h N_v^hP^{-h}_v(\cos\theta)e^{jh\phi}, &h< 0,
\end{cases}
\end{equation}
where $N_v^h= \sqrt{\frac{(2v+1)(v-|h|)!}{4\pi(v+|h|)!}}$ is the normalized factor, $P_v^h(x)$ is the associated Legendre function of degree $v$ and order $h$, which is expressed as 
$
	P_v^h(x)=(1-x)^{\frac{h}{2}}\frac{d^h P_v(x)}{dx^h},
$
with $ P_v(x)$ being the Legendre function \cite{costa2010unified}. Note that for any given $v$ and $h$, the spherical harmonic basis functions can be directly determined by \eqref{eq_spherical_harmonic_basis}. Therefore, with the assistance of the SHOD method, it is theoretically possible to synthesize any desired radiation pattern by carefully tuning the harmonic coefficients.
However, employing an infinite number of spherical harmonic basis functions for radiation pattern synthesis would lead to significant computational complexity in the design. To mitigate this issue, a truncation method is adopted in \eqref{eq_SHOD}. Specifically, the expansion is truncated to the first $I = V^2 + V + 1$ terms, meaning that only the first $I$ spherical harmonic basis functions are used to approximate the radiation patterns. Thus, the approximated pattern can be written as
\begin{align}\label{eq_SHOD_truncation}
p_m\left(\theta, \phi\right)\approx \tilde{p}_m\left(\theta, \phi\right)&=\sum_{v=0}^{V}\sum_{h=-v}^{v}\omega_{m, v, h}Y^h_{v}\left(\theta, \phi\right),\notag\\
&=\sum_{i=1}^{I}\omega_{m, i}Y_{i}\left(\theta, \phi\right)=\boldsymbol{\omega}_m^H\boldsymbol{\eta}\left(\theta, \phi\right),
\end{align}
where $\omega_{m, i}=\omega_{m, v, h}$ and $Y_{i}\left(\theta, \phi\right)=Y^h_{v}\left(\theta, \phi\right)$ for any $i=v^2+v+h+1, ~ i\in\mathcal{I}\triangleq\{1, 2,\dots, I\}$, $v\in[0, V]$, $h\in[-v, v]$, $\boldsymbol{\omega}_m=[\omega_{m, 1}, \dots, \omega_{m, I}]^H\in\mathbb{C}^{I\times 1}$, $\boldsymbol{\eta}\left(\theta, \phi\right)=[Y_{1}\left(\theta, \phi\right), \dots, Y_{I}\left(\theta, \phi\right)]^T$. Also, for $i, s\in\mathcal{I}$, we have
\begin{align}\label{eq_spherical_harmonic_orx}
	\iint Y_i\left(\theta, \phi\right)	Y_s\left(\theta, \phi\right)\sin\theta d\theta d\phi=\begin{cases}
		1, i=s,\\
		0, i\neq s.
	\end{cases}
\end{align}

Recall that the design of $\boldsymbol{\omega}_m$ directly determines the shape of the element's radiation pattern. For this reason, we will focus on optimizing these weight coefficients to achieve the desired antenna radiation patterns. It is worth noting that when designing the fluid elements' radiation patterns, an energy constraint should be imposed on the radiation pattern gain of each element. Specifically, this constraint is expressed as $\iint |\tilde{p}_m\left(\theta, \phi\right)|^2\sin\theta d\theta d\phi=4\pi$ \cite{zheng2025tri}. Utilizing the result in \eqref{eq_spherical_harmonic_orx}, the energy constraint can be equivalently rewritten as $\left\|\boldsymbol{\omega}_m\right\|^2_2=4\pi$. This transformation simplifies the design process by relating the constraint directly to the weight coefficients, facilitating more efficient optimization. To validate the effectiveness of the SHOD method, we reconstruct the 3GPP 38.901 radiation pattern as specified in \cite{3GPP} and illustrate the synthesized 3GPP 38.901 patterns under various truncation lengths $I$, along with the corresponding harmonic coefficients in Fig.~\ref{fig:SHOD_method_verification}. Note that the normalized mean square error (NMSE) is employed to evaluate the reconstruction accuracy of the patterns. According to the obtained results, it is observed that the NMSE consistently decreases with an increasing truncation length $I$, indicating improved reconstruction accuracy. Furthermore, the harmonic coefficient distributions indicate that the majority of the pattern energy is concentrated in the low-frequency components (i.e., small $v$), which demonstrates that a reasonable truncation operation can effectively approximate the original radiation pattern. As a consequence, the channel links between the FRIS and the BS/UE $k$, i.e., $\mathbf{H}_\mathrm{BR}$ and $\mathbf{h}_{\mathrm{r}k}$, can be re-expressed, respectively, as 
\begin{equation}\label{eq_channel_rewritten}
\mathbf{H}_{\mathrm{BR}}=\boldsymbol{\Omega}^H\mathbf{A},~\mbox{and}~\mathbf{h}_{\mathrm{r}k}=\boldsymbol{\Omega}^H\mathbf{b}_k,
\end{equation}
in which $\mathbf{A}=\mathbf{A}^\mathrm{LoS}+\mathbf{A}^\mathrm{NLoS}$, $\mathbf{b}_k=\mathbf{b}_k^\mathrm{LoS}+\mathbf{b}_k^\mathrm{NLoS}$, $\boldsymbol{\Omega}=\operatorname{Blkdiag}\{\boldsymbol{\omega}_1, \dots, \boldsymbol{\omega}_m, \dots, \boldsymbol{\omega}_M\}$, and 
\begin{equation}
\left\{\begin{aligned}
\mathbf{A}^\mathrm{LoS}&=g^\mathrm{LoS}_{\mathrm{br}}\gamma_1\rho(d^\mathrm{LoS}_\mathrm{BR})\boldsymbol{\Upsilon}\left(\theta_\mathrm{r}^\mathrm{LoS}, \phi_\mathrm{r}^\mathrm{LoS}\right)\mathbf{a}_\mathrm{R}\left(\theta_\mathrm{r}^\mathrm{LoS}, \phi_\mathrm{r}^\mathrm{LoS}\right)\\
&\times\mathbf{a}^H_\mathrm{B}\left(\theta_\mathrm{t, B}^\mathrm{LoS}, \phi_\mathrm{t, B}^\mathrm{LoS}\right)~\mbox{with}~\boldsymbol{\Upsilon}=\mathbf{I}_{M}\otimes\boldsymbol{\eta},\\
\mathbf{A}^\mathrm{NLoS}&=\gamma_1\sum_{l=1}^{L}\rho(d^l_\mathrm{BR})g^l_{\mathrm{br}}\boldsymbol{\Upsilon}\left(\theta_\mathrm{r}^{l}, \phi_\mathrm{r}^{l}\right)\mathbf{a}_\mathrm{R}\left(\theta_\mathrm{r}^{l}, \phi_\mathrm{r}^{l}\right)\\
&\times\mathbf{a}^H_\mathrm{B}\left(\theta_\mathrm{t, B}^{l}, \phi_\mathrm{t, B}^{l}\right),\\
\mathbf{b}_k^\mathrm{LoS}&=\gamma_2g_{\mathrm{r}k}^{\mathrm{LoS}}\rho(d^{\mathrm{LoS}}_{\mathrm{r}k})\boldsymbol{\Upsilon}\left(\theta_{\mathrm{t}, k}^{\mathrm{LoS}}, \phi_{\mathrm{t}, k}^{\mathrm{LoS}}\right)\mathbf{a}_\mathrm{R}\left(\theta_{\mathrm{t}, k}^{\mathrm{LoS}}, \phi_{\mathrm{t}, k}^{\mathrm{LoS}}\right),\\
\mathbf{b}_k^\mathrm{NLoS}&=\gamma_2\sum_{z=1}^Z \rho(d^z_{\mathrm{r}k})g_{\mathrm{r}k}^z\boldsymbol{\Upsilon}\left(\theta_{\mathrm{t}, k}^{z}, \phi_{\mathrm{t}, k}^{z}\right)\mathbf{a}_\mathrm{R}\left(\theta_{\mathrm{t}, k}^{z}, \phi_{\mathrm{t}, k}^{z}\right).
\end{aligned}\right.
\end{equation}

\vspace{-2mm}
\section{Problem Formulation and Algorithm Design}\label{sec:S4}
\subsection{Problem Formulation}
Under the multiuser communication scenario, we formulate an optimization problem aimed at maximizing the weighted sum rate of the UEs. This formulation takes into account the power budget at the BS, and the energy constraints of the fluid elements' radiation patterns. That is,
\begin{subequations}\label{eq_orig_opti}
	\begin{align}
		&\max _{\boldsymbol{\Upsilon}}~~~ \sum_{k=1}^{K}\epsilon_k R_k,\notag \\
		&\quad\text { s.t. }\sum_{k=1}^{K}\left\|\mathbf{w}_k\right\|_2^2\leq P_\mathrm{tmax}, \label{eq_orig_opti_1}\\
		&\quad\qquad \left\|\boldsymbol{\omega}_m\right\|_2^2=4\pi, \forall m\in\mathcal{M}, \label{eq_orig_opti_2}
	\end{align}
\end{subequations} 
where $\boldsymbol{\Upsilon}\triangleq\{\{\boldsymbol{\omega}_m\}_{m=1}^M, \{\mathbf{w}_k\}_{k=1}^K\}$, $\epsilon_k$ represents the weight coefficient of the $k$-th UE, playing a crucial role in determining user priority and fairness in communication systems, $P_\mathrm{tmax}$ is the power budget at the BS. Note that the constraints in \eqref{eq_orig_opti_2} are the equivalent forms of the patterns' energy constraints. Addressing this optimization problem is challenging due to the non-concave objective function and the equality constraints imposed by the harmonic coefficients. To tackle this issue, an alternating optimization strategy is adopted to decompose the original problem into two subproblems: the active beamforming subproblem and the pattern design subproblem. These two subproblems are iteratively solved using the MMSE method and the RCG method, respectively. The following subsection details the design process of the proposed algorithm.

\subsection{Algorithm Design}
\subsubsection{Active Beamforming Design} Given harmonic coefficients, the problem can be simplified to
\begin{equation}\label{eq_active_beamforming_sub}
\max _{\{\mathbf{w}_k\}_{k=1}^K}~ \sum_{k=1}^{K}\epsilon_k R_k~~\mbox{s.t.}~~\eqref{eq_orig_opti_1}. 
\end{equation}
Solving this subproblem directly is still difficult due to the non-concavity of the objective function w.r.t.~the active beamforming variables $\{\mathbf{w}_k\}_{k=1}^K$. To overcome this challenge, the MMSE method is adopted to transform $R_k$ into
\begin{equation}
R_k=\max_{D_k, u_k}\underbrace{\frac{\ln(D_k)-D_kE_K(u_k, \{\mathbf{w}_k\}_{k=1}^K)+1}{\ln2}}_{\tilde{R}_k},
\end{equation}
where 
\begin{multline}
E_k=u_k^H\left(\sum_{j\neq k}^{K}|\mathbf{h}_{\mathrm{r}k}^H\mathbf{H}_\mathrm{BR}\mathbf{w}_j|^2+\sigma_k^2\right)u_k\\
+\left(u_k^H\mathbf{h}_{\mathrm{r}k}^H\mathbf{H}_\mathrm{BR}\mathbf{w}_k-1\right)\left(u_k^H\mathbf{h}_{\mathrm{r}k}^H\mathbf{H}_\mathrm{BR}\mathbf{w}_k-1\right)^H.
\end{multline}

Additionally, it can be readily shown that $\tilde{R}_k$ is a concave function w.r.t.~the variables $D_k$ and $u_k$. Therefore, the optimal solution can be obtained by equating the partial derivatives of $R_k$ w.r.t.~$D_k$ and $u_k$ to zero, which gives
\begin{equation}\label{eq_auxiliary}
\left\{\begin{aligned}
D_k^\mathrm{opt}&=E_k^{-1},\\
u_k^\mathrm{opt}&=\frac{\mathbf{h}_{\mathrm{r}k}^H\mathbf{H}_\mathrm{BR}\mathbf{w}_k}{\sum_{j=1}^{K}|\mathbf{h}_{\mathrm{r}k}^H\mathbf{H}_\mathrm{BR}\mathbf{w}_j|^2+\sigma_k^2}.
\end{aligned}\right.
\end{equation}

In the $(t+1)$-th iteration, the concave lower bound of $R_k$ can be obtained with the given $D_k^{(t)}$ and $u_k^{(t)}$, i.e., $\tilde{R}_k(D^{(t)}_k,u^{(t)} )$. In this case, in the $(t+1)$-th iteration,  the subproblem \eqref{eq_active_beamforming_sub} can be transformed into
\begin{equation}\label{eq_active_beamforming_sub_trans}
\max _{\{\mathbf{w}_k\}_{k=1}^K} \sum_{k=1}^{K}\sum_{j=1}^{K}-\mathbf{w}_j^H\mathbf{F}_k\mathbf{w}_j+2\operatorname{R}\left(\mathbf{f}_k^H\mathbf{w}_k\right)+G_k~~\mbox{s.t.}~~\eqref{eq_orig_opti_1}, 
\end{equation}
where $\mathbf{f}_k=\frac{\epsilon_k D_k^{(t)}}{\ln2}\mathbf{H}_\mathrm{BR}^H\mathbf{h}_{\mathrm{r}k}u^{(t)}_k$, $\mathbf{F}_k=\frac{\mathbf{f}_k\mathbf{f}^H_k\ln2 }{\epsilon_k D_k^{(t)}}$, and $G_k=\epsilon_k\frac{\ln D_k^{(t)}-D_k^{(t)}\left(1+\sigma_k^2|u_k^{(t)}|^2\right)}{\ln2}$.

Note that \eqref{eq_active_beamforming_sub_trans} is convex and can be solved using standard solvers such as CVX. However, employing such solvers often incurs high computational complexity. To mitigate this issue, we instead adopt the Lagrange multiplier method as a more efficient alternative. Specifically, the Lagrange function associated with the optimization problem above is given by
\begin{multline}\label{eq_Larange}
\mathscr{L}=\sum_{k=1}^{K}\sum_{j=1}^{K}\mathbf{w}_j^H\mathbf{F}_k\mathbf{w}_j-2\operatorname{R}\left(\mathbf{f}_k^H\mathbf{w}_k\right)-G_k\\
+\upsilon\left(\sum_{k=1}^{K}\left\|\mathbf{w}_k\right\|_2^2-P_\mathrm{tmax}\right),
\end{multline}
where $\upsilon\geq 0$ denotes the Lagrange multiplier (dual variable). Thus, the partial derivative of the Lagrange function w.r.t.~active beamforming can be derived as
\begin{equation}
\frac{\partial \mathscr{L}}{\partial \mathbf{w}_k}=\sum_{j=1}^{K}2\mathbf{F}_j\mathbf{w}_k-2\mathbf{f}_k+2\upsilon\mathbf{w}_k, k\in\mathcal{K}.
\end{equation}
Letting $\frac{\partial \mathscr{L}}{\partial \mathbf{w}_k}=0$, the optimal $\mathbf{w}_k$ associated with the corresponding dual variable can be obtained as \begin{equation}\label{eq_active_beamforming_opt_dual}
\mathbf{w}^\mathrm{opt}_k(\upsilon)=(\mathbf{F}+\upsilon\mathbf{I}_{N_\mathrm{t}\times N_\mathrm{t}})^{\dagger}\mathbf{f}_k,
\end{equation}
where $\mathbf{F}=\sum_{j=1}^{K}\mathbf{F}_j$.

It is important to note that \eqref{eq_active_beamforming_sub_trans} satisfies Slater's condition, which ensures strong duality. Hence, the optimal solution to \eqref{eq_active_beamforming_sub_trans} can be obtained by determining the optimal dual variable. In particular, the optimal dual variable $\upsilon^\mathrm{opt}$ must satisfy the complementary slackness condition, given by
\begin{equation}
\upsilon^\mathrm{opt}\left(\sum_{k=1}^{K}\left\|\mathbf{w}^\mathrm{opt}_k(\upsilon^\mathrm{opt})\right\|_2^2-P_\mathrm{tmax}\right)=0.
\end{equation}

If $\sum_{k=1}^{K}\left\|\mathbf{w}^\mathrm{opt}_k(\upsilon^\mathrm{opt})\right\|_2^2< P_\mathrm{tmax}$ holds, $\upsilon^\mathrm{opt}=0$ and the optimal solution of \eqref{eq_active_beamforming_sub_trans} is $\mathbf{w}^\mathrm{opt}_k=\mathbf{F}^{\dagger}\mathbf{f}_k$. Otherwise, 
\begin{equation}\label{eq_dual_variable_solve}
P(\upsilon^\mathrm{opt})=\sum_{k=1}^{K}\left\|\mathbf{w}^\mathrm{opt}_k(\upsilon^\mathrm{opt})\right\|_2^2=P_\mathrm{tmax}.
\end{equation}
In fact, solving \eqref{eq_dual_variable_solve} is challenging due to the complicated form of $\mathbf{w}^\mathrm{opt}_k$. Moreover, the uncertainty regarding the monotonicity of the function $P(\upsilon)$ makes it difficult to apply conventional numerical methods for solving the equation. We address this issue by the following theorem.

\begin{theorem}\label{th_1}
$P(\upsilon)$ is monotonic-decreasing w.r.t.~$\upsilon$.
\end{theorem}

\begin{proof}
Firstly, we define two variables $\upsilon_1$ and $\upsilon_2$ and  $\upsilon_1>\upsilon_2\geq 0$. Then we can obtain the optimal active beamforming for \eqref{eq_Larange} associated with $\upsilon_1$ and $\upsilon_2$, i.e., $\{\mathbf{w}^\mathrm{opt}_k(\upsilon_1)\}_{k=1}^K$ and $\{\mathbf{w}^\mathrm{opt}_k(\upsilon_2)\}_{k=1}^K$. For these two solutions, we have 
\begin{align}
&\mathscr{L}(\{\mathbf{w}^\mathrm{opt}_k(\upsilon_1)\}_{k=1}^K, \upsilon_1)\leq\mathscr{L}(\{\mathbf{w}^\mathrm{opt}_k(\upsilon_2)\}_{k=1}^K, \upsilon_1),\\
&\mathscr{L}(\{\mathbf{w}^\mathrm{opt}_k(\upsilon_2)\}_{k=1}^K, \upsilon_2)\leq\mathscr{L}(\{\mathbf{w}^\mathrm{opt}_k(\upsilon_1)\}_{k=1}^K, \upsilon_2),
\end{align}
By adding these two inequalities, we can derive the condition
$
(\upsilon_1-\upsilon_2)P(\upsilon_1)\leq(\upsilon_1-\upsilon_2)P(\upsilon_2). 
$
Considering $\upsilon_1>\upsilon_2$, we have $P(\upsilon_1)\leq P(\upsilon_2)$. Therefore, we can demonstrate that $P(\upsilon)$ is a monotonically decreasing function w.r.t.~$\upsilon$.
\end{proof}

\vspace{-2mm}
\begin{center}
\begin{small}
\begin{tabular}{p{8.5cm}}
\toprule[1pt]
\textbf{Algorithm 1:}  Algorithm Based on Bisection Search Method for Solving Active Beamforming Subproblem \eqref{eq_active_beamforming_sub_trans}   \\
\midrule[.5pt]
1: Initialize the tolerance accuracy $\epsilon$ and the bound of the\\
\quad dual variable $\upsilon_\mathrm{up}$ and $\upsilon_\mathrm{low}$.\\
2: If condition $\sum_{k=1}^{K}\left\|\mathbf{w}^\mathrm{opt}_k(\upsilon^\mathrm{opt})\right\|_2^2< P_\mathrm{tmax}$ holds,  $\mu^\mathrm{opt}$\\\quad $=0$ and $\mathbf{w}^\mathrm{opt}_k=\mathbf{F}^{\dagger}\mathbf{f}_k$. Otherwise, go to Steps 3-6 to \\\quad solve $\upsilon^\mathrm{opt}$ in \eqref{eq_dual_variable_solve}.\\
3: \textbf{Repeat} \\
4: \quad Calculate $\upsilon^\mathrm{opt}=\frac{\upsilon_\mathrm{up}+\upsilon_\mathrm{low}}{2}$ and $P(\upsilon^\mathrm{opt})$.\\
5: \quad if $P(\upsilon^\mathrm{opt})\geq P_\mathrm{tmax}$, $ \upsilon_\mathrm{low}\leftarrow\upsilon^\mathrm{opt}$. Otherwise, $\upsilon_\mathrm{up}\leftarrow$\\\qquad $\upsilon^\mathrm{opt}$. \\
6: \textbf{Until} $|\upsilon_\mathrm{up}-\upsilon_\mathrm{low}|\leq \epsilon$\\
7: \textbf{Output:} the optimal dual variable $\upsilon^\mathrm{opt}$ and the optimal\\\quad active beamforming $\mathbf{w}^\mathrm{opt}_k(\upsilon^\mathrm{opt})$ through \eqref{eq_active_beamforming_opt_dual}. \\
\bottomrule[1pt]
\end{tabular}
\end{small}
\end{center}

According to the result in Theorem \ref{th_1}, the bisection search method can be employed to solve \eqref{eq_dual_variable_solve}, and obtain the unique optimal solution. Subsequently, the optimal active beamforming can be achieved using \eqref{eq_active_beamforming_opt_dual}.  Algorithm 1 summarizes the detailed solving process for the problem \eqref{eq_active_beamforming_sub_trans}.

\subsubsection{Pattern Design} After getting the active beamforming, we turn our attention to the design of the harmonic coefficients $\boldsymbol{\omega}_m, ~m\in\mathcal{M}$, which are related to the pattern configuration. Specifically, the subproblem corresponding to the pattern design can be formulated as
\begin{align}\label{eq_pattern_design_trans}
&\max _{\boldsymbol{\Omega}}~~ f(\boldsymbol{\Omega})=\sum_{k=1}^{K}\epsilon_k R_k,\notag \\
&~\text { s.t. }~\left[\boldsymbol{\Omega}^H\boldsymbol{\Omega}\right]_{m,m}=4\pi,~m\in\mathcal{M},
\end{align}
with the rewritten form of $\mathbf{H}_\mathrm{BR}$ and $\mathbf{h}_{\mathrm{r}k}$ in \eqref{eq_channel_rewritten}. Moreover, $R_k$ can be re-expressed as	
\begin{multline}
R_k=\log_2\left(\operatorname{Tr}\left(\boldsymbol{\Omega}\boldsymbol{\Omega}^H\widehat{\mathbf{A}}\boldsymbol{\Omega}\boldsymbol{\Omega}^H\mathbf{B}_k\right)+\sigma_k^2\right)\\
-\log_2\left(\operatorname{Tr}\left(\boldsymbol{\Omega}\boldsymbol{\Omega}^H\tilde{\mathbf{A}}_k\boldsymbol{\Omega}\boldsymbol{\Omega}^H\mathbf{B}_k\right)+\sigma_k^2\right),
\end{multline}
in which we have defined $\widehat{\mathbf{A}}=\mathbf{A}\sum_{j=1}^{K}\mathbf{w}_j\mathbf{w}_j^H\mathbf{A}^H$, $\mathbf{B}_k=\mathbf{b}_k\mathbf{b}_k^H$, and $\tilde{\mathbf{A}}_k=\mathbf{A}\sum_{j\neq k}^{K}\mathbf{w}_j\mathbf{w}_j^H\mathbf{A}^H$.

Note that the existence of the equality constraints in \eqref{eq_pattern_design_trans} characterizes the feasible set of \eqref{eq_pattern_design_trans} as an oblique manifold space, given by $\mathscr{O} \triangleq \{ \boldsymbol{\Omega} \in \mathbb{C}^{IM \times M} \mid [\boldsymbol{\Omega}^H \boldsymbol{\Omega}]_{m,m} = 4\pi,~ m \in \mathcal{M}\}$. To address this, we adopt the RCG algorithm \cite{xiao2025multi} for efficiently obtaining a sub-optimal solution of $\boldsymbol{\Omega}$. In this algorithm, the Riemannian gradient is first computed by orthogonally projecting the Euclidean gradient onto the tangent space of the manifold, defined as $\mathcal{T}_{\mathscr{O}} \triangleq \{ \mathbf{T} \in \mathbb{C}^{IM \times M} \mid [\mathbf{T}^H \boldsymbol{\Omega}]_{m, m} = 0, m \in \mathcal{M} \}$. The Riemannian gradient is used to determine the direction of steepest descent on the manifold for updating the next iterate. Since the updated point resides in the tangent space $\mathcal{T}_{\mathscr{O}}$, a retraction operation is applied to map the point back onto the manifold. Specifically, in the $(e+1)$-th iteration, the Riemannian gradient of the objective function w.r.t.~$\boldsymbol{\Omega}$ can be derived as
\begin{align}\label{eq_Riemannian_gradient}
\operatorname{Rgrad}f^{(e+1)}&=\mathscr{M}_{\mathcal{T}_{\mathscr{O}}}\left(\nabla f\left(\boldsymbol{\Omega}^{(e)}\right)\right)\notag\\
&=\nabla f\left(\boldsymbol{\Omega}^{(e)}\right)-\frac{\boldsymbol{\Omega}^{(e)}}{4\pi}\operatorname{Ddiag}\left(\mathbf{O}^{(e)}\right),
\end{align}
where $\mathbf{O}^{(e)}=\operatorname{R}\Big(\left(\boldsymbol{\Omega}^{(e)}\right)^H\nabla f\left(\boldsymbol{\Omega}^{(e)}\right)\Big)$, $\mathscr{M}_{\mathcal{T}_{\mathscr{O}}}(\cdot)$ denotes the projecting operator, and $\nabla f\left(\boldsymbol{\Omega}\right)$ is the gradient of the objective function w.r.t.~$\boldsymbol{\Omega}$ in the Euclidean space, given as
\begin{equation}\label{eq_Euclidean_gradient}
\nabla f\left(\boldsymbol{\Omega}\right)=\sum_{k=1}^{K}\frac{2\epsilon_k}{\ln2}\left(\boldsymbol{\Upsilon}_{1, k}-\boldsymbol{\Upsilon}_{2,k}\right)\circ\mathbf{M},
\end{equation}
where $\mathbf{M}=\mathbf{I}_M\otimes \mathbf{1}_{I\times 1}$ denotes the mask block diagonal matrix, $\boldsymbol{\Upsilon}_{1, k}$ and $\boldsymbol{\Upsilon}_{2, k}$ are, respectively, given by
\begin{equation}
\left\{\begin{aligned}
\boldsymbol{\Upsilon}_{1,k}&=\frac{\widehat{\mathbf{A}}\boldsymbol{\Omega}\boldsymbol{\Omega}^H\mathbf{B}_k\boldsymbol{\Omega}+\mathbf{B}_k\boldsymbol{\Omega}\boldsymbol{\Omega}^H\widehat{\mathbf{A}}\boldsymbol{\Omega}}{\operatorname{Tr}\left(\boldsymbol{\Omega}\boldsymbol{\Omega}^H\widehat{\mathbf{A}}\boldsymbol{\Omega}\boldsymbol{\Omega}^H\mathbf{B}_k\right)+\sigma_k^2}\\
&=\frac{\mathbf{A}\mathbf{W}\mathbf{H}^H_\mathrm{BR}\mathbf{h}_{\mathrm{r}k}\mathbf{h}^H_{\mathrm{r}k}+\mathbf{b}_k\mathbf{h}^H_{\mathrm{r}k}\mathbf{H}_\mathrm{BR}\mathbf{W}\mathbf{H}^H_\mathrm{BR}}{\sum_{j=1}^{K}\left|\mathbf{h}_{\mathrm{r}k}^H\mathbf{H}_\mathrm{BR}\mathbf{w}_j\right|^2+\sigma_k^2},\\
\boldsymbol{\Upsilon}_{2,k}&=\frac{\tilde{\mathbf{A}}_k\boldsymbol{\Omega}\boldsymbol{\Omega}^H\mathbf{B}_k\boldsymbol{\Omega}+\mathbf{B}_k\boldsymbol{\Omega}\boldsymbol{\Omega}^H\tilde{\mathbf{A}}_k\boldsymbol{\Omega}}{\operatorname{Tr}\left(\boldsymbol{\Omega}\boldsymbol{\Omega}^H\tilde{\mathbf{A}}_k\boldsymbol{\Omega}\boldsymbol{\Omega}^H\mathbf{B}_k\right)+\sigma_k^2}\\
&=\frac{\mathbf{A}\mathbf{W}_{-k}\mathbf{H}^H_\mathrm{BR}\mathbf{h}_{\mathrm{r}k}\mathbf{h}^H_{\mathrm{r}k}+\mathbf{b}_k\mathbf{h}^H_{\mathrm{r}k}\mathbf{H}_\mathrm{BR}\mathbf{W}_{-k}\mathbf{H}^H_\mathrm{BR}}{\sum_{j\neq k}^{K}\left|\mathbf{h}_{\mathrm{r}k}^H\mathbf{H}_\mathrm{BR}\mathbf{w}_j\right|^2+\sigma_k^2},
\end{aligned}\right.
\end{equation}
with $\mathbf{W}=\sum_{j=1}^{K}\mathbf{w}_j\mathbf{w}^H_j$ and $\mathbf{W}_{-k}=\sum_{j\neq k}^{K}\mathbf{w}_j\mathbf{w}^H_j$.

Thus, the iterative direction of the next point is found by
\begin{equation}\label{eq_search_direction}
\mathcal{D}^{(e+1)}=\operatorname{Rgrad}f^{(e+1)}+\zeta^{(e+1)}_\mathrm{PR}\mathscr{M}_{\mathcal{T}_{\mathscr{O}}}(\mathcal{D}^{(e)}),
\end{equation}
where $\zeta^{(e+1)}_\mathrm{PR}$ denotes the Polak-Ribiere parameter in the $(e+1)$-th iteration, which is determined by the expression \eqref{eq_Polak-Ribiere_parameter}, shown at the top of the next page.
\begin{figure*}[t!]
\begin{equation}\label{eq_Polak-Ribiere_parameter}
\zeta^{(e+1)}_\mathrm{PR}=\operatorname{R}\left(\frac{\operatorname{Tr}\left(\left(\operatorname{Rgrad}f^{(e+1)}\right)^H\left(\operatorname{Rgrad}f^{(e+1)}-\mathscr{M}_{\mathcal{T}_{\mathscr{O}}}\left(\operatorname{Rgrad}f^{(e)}\right)\right)\right)}{\operatorname{Tr}\left(\left(\operatorname{Rgrad}f^{(e)}\right)^H\operatorname{Rgrad}f^{(e)}\right)}\right)
\end{equation}
\hrule
\end{figure*}
Hence, the updated rule of the $(e+1)$-th iteration point can be given by
\begin{equation}
\boldsymbol{\Omega}^{(e+1)}=\mathscr{R}\left(\boldsymbol{\Omega}^{(e)}+\psi^{(e+1)}\mathcal{D}^{(e+1)}\right),
\end{equation}
where $\mathscr{R}(\cdot)$ denotes the retraction operator, which projects the updated point onto the manifold, $\psi^{(e+1)}$ denotes the Armijo step in the $(e+1)$-th iteration.
The procedure of addressing \eqref{eq_pattern_design_trans} utilizing RCG is presented as Algorithm 2.

\subsection{Complexity and Convergence Analysis}
The overall procedure of the proposed iterative algorithm for solving the optimization problem in \eqref{eq_orig_opti} is summarized in Algorithm 3. The algorithm is considered to have converged when the difference in objective function values between two consecutive iterations falls below a predefined threshold $\tilde{\epsilon}$.

The computational complexity of the proposed algorithm primarily arises from solving the subproblems \eqref{eq_active_beamforming_sub_trans} and \eqref{eq_pattern_design_trans}. Specifically, for the active beamforming subproblem \eqref{eq_active_beamforming_sub_trans}, an efficient solution is obtained by applying the Lagrange multiplier method in combination with the bisection search technique. During this process, the dominant computational burden lies in calculating the optimal active beamforming vector associated with the dual variable, as given in \eqref{eq_active_beamforming_opt_dual}. Consequently, the total computational complexity of Algorithm 1 is given by $\mathscr{O}_1=\mathcal{O}(T_1(N_\mathrm{t}^3+N_\mathrm{t}^2))$ with $T_1$ denoting the total number of iterations of Algorithm 1. In terms of the subproblem \eqref{eq_pattern_design_trans}, the RCG method is leveraged to effectively obtain the suboptimal harmonic coefficients. The main computing complexity focuses on the calculation of the Euclidean gradient in \eqref{eq_Euclidean_gradient}. Thus, the computational complexity of Algorithm 2 can be calculated as $\mathscr{O}_2=\mathcal{O}(T_2(IM^3N_\mathrm{t}^2))$, where $T_2$ is the total number of iterations of Algorithm 2. As a result, the overall computational complexity of the proposed iterative algorithm can be estimated as $\mathscr{O}_\mathrm{tol}=T(\mathscr{O}_1+\mathscr{O}_2)$ with $T$ being the total number of iterations of Algorithm 3.

The convergence of the proposed algorithm is theoretically guaranteed, as the alternating optimization framework ensures that the objective function value does not decrease with each iteration. Specifically, at every step of the iteration, the updated solution yields an objective value that is no worse than that of the previous iteration. Furthermore, the convergence behavior of the proposed algorithm will be thoroughly evaluated through extensive simulations in Section \ref{sec:S5}.

\vspace{-4mm}
\begin{center}
\begin{small}
\begin{tabular}{p{8.5cm}}
\toprule[1pt]
\textbf{Algorithm 2:}  RCG algorithm for Solving Pattern Design Subproblem \eqref{eq_pattern_design_trans}   \\
\midrule[0.5pt]
1: Initialize variable $\boldsymbol{\Omega}^{(0)}$ and tolerance accuracy $\hat{\epsilon}$; Set the\\\quad iteration index $e=0$ \\
2: \textbf{Repeat} \\
3: \quad Calculate objective function $f(\boldsymbol{\Omega}^{(e)})$, Euclidean gradi-\\\qquad ent $\nabla f(\boldsymbol{\Omega}^{(e)})$ and Riemannian gradient $\operatorname{Rgrad}f^{(e+1)}$\\\qquad using \eqref{eq_Euclidean_gradient} and \eqref{eq_Riemannian_gradient}, respectively.\\
4: \quad Calculate the Polak-Ribiere parameter $\zeta^{(e+1)}_\mathrm{PR}$ and sea-\\\qquad rch direction $\mathcal{D}^{(e+1)}$ using \eqref{eq_Polak-Ribiere_parameter} and \eqref{eq_search_direction}, respectively.\\
5: \quad Calculate the Armijo step $\psi^{(e+1)}$, 
Update $\boldsymbol{\Omega}^{(e+1)}$ and \\\qquad the corresponding objective function  $f(\boldsymbol{\Omega}^{(e+1)})$; Let \\\qquad  $e\leftarrow e+1$.\\
6: \textbf{Until} $|f(\boldsymbol{\Omega}^{(e+1)})-f(\boldsymbol{\Omega}^{(e)})|\leq \hat{\epsilon}$\\
7: \textbf{Output:} the optimal $\boldsymbol{\Omega}$ and the corresponding harmonic\\\quad coefficients $\boldsymbol{\omega}_m, ~m\in\mathcal{M}$. \\
\bottomrule[1pt]
\end{tabular}
\end{small}
\end{center}

\vspace{-4mm}
\begin{center}
\begin{small}
\begin{tabular}{p{8.5cm}}
\toprule[1pt]
\textbf{Algorithm 3:}  The Proposed Iterative Algorithm for Addressing Optimization Problem \eqref{eq_orig_opti}   \\
\midrule[0.5pt]
1: Initialize variable $(\{\mathbf{w}^{(0)}_k\}_{k=1}^K, \{\boldsymbol{\omega}_m^{(0)}\}_{m=1}^M)$ and tolerance\\\quad accuracy $\tilde{\epsilon}$; Set the iteration index $t=0$ \\
2: \textbf{Repeat} \\
3: \quad Calculate $D_k^{(t)}$ and $u_k^{(t)}$ with the given $\big(\{\mathbf{w}^{(t)}_k\}_{k=1}^K,$\\\qquad$ \{\boldsymbol{\omega}_m^{(t)}\}_{m=1}^M\big)$ by leveraging \eqref{eq_auxiliary}.\\
4: \quad Solving the problem \eqref{eq_active_beamforming_sub_trans} considering the obtained\\\qquad $D_k^{(t)}$ and $u_k^{(t)}$ by Algorithm 1; Update $\{\mathbf{w}^{(t+1)}_k\}_{k=1}^K$\\\qquad with the obtained solution. \\
5: \quad Solving problem \eqref{eq_pattern_design_trans} with the obtained $\{\mathbf{w}^{(t+1)}_k\}_{k=1}^K$\\\qquad utilizing Algorithm 2 and update $\{\boldsymbol{\omega}_m^{(t+1)}\}_{m=1}^M$ with\\\qquad the achieved solution.\\
6: \textbf{Until} the gap in objective function values between adj-\\\quad acent iterations falls below $\tilde{\epsilon}$\\
7: \textbf{Output:} the optimal $(\{\mathbf{w}_k\}_{k=1}^K, \{\boldsymbol{\omega}_m\}_{m=1}^M)$. \\
\bottomrule[1pt]
\end{tabular}
\end{small}
\end{center}

\vspace{-2mm}
\section{Simulation Results}\label{sec:S5}
This section reports numerical results that attempt to answer the following three questions:
\begin{itemize}
\item[Q1)] What is the extent of performance improvement enabled by pattern-reconfigurable FRIS in multiuser settings?
\item[Q2)] What is the fundamental mechanism by which the shape of the radiation pattern affects system performance?
\item[Q3)] To what extent can the pattern-reconfigurable FRIS reduce the number of RIS elements and antennas compared to a conventional fixed-pattern RIS?
\end{itemize}

\vspace{-2mm}
\subsection{Setup}
To understand the effectiveness of the proposed pattern-reconfigurable FRIS (which is referred to as the \textbf{FRIS+MMSE} scheme), simulation results are presented and compared against the following three benchmark schemes:
\begin{itemize}
\item \textbf{FRIS+ZF scheme}---This scheme uses the zero-forcing (ZF) principle to design the active beamforming vectors, $\{\mathbf{w}_k\}_{k=1}^K$, which are found as
\begin{equation}
\mathbf{w}_k=\sqrt{P_k}\frac{\mathbf{W}(:,k)}{\|\mathbf{W}(:,k)\|_2},~ k\in\mathcal{K},
\end{equation}
in which $\mathbf{W}=\mathbf{X}^H\left(\mathbf{X}\mathbf{X}^H\right)^{-1}$ is obtained based on $\mathbf{X}=[\mathbf{h}_{\mathrm{r}1},\dots, \mathbf{h}_{\mathrm{r}k},\dots, \mathbf{h}_{\mathrm{r}K}]^H\mathbf{H}_\mathrm{BR}$, and $P_k$ denotes the allocated power of the BS for the $k$-th UE, which will be carefully optimized by the CVX solver.
\item \textbf{Traditional RIS with 38.901 pattern scheme}---In this scheme, each element's pattern is chosen to be the 3GPP  38.901 pattern. Additionally, the passive beamforming capability of the conventional RIS is utilized to manipulate the incident signals, with the design of the beamforming coefficients following the method proposed in \cite{xiao2025multi}.
\item \textbf{Traditional RIS with isotropic pattern scheme}---Here, the isotropic pattern is adopted for the RIS elements, which is commonly used in most existing studies on traditional RIS-assisted communication systems. Similarly, the passive beamforming capability of the conventional RIS is employed to manipulate the incident signals.
\end{itemize}

Unless stated otherwise, the system parameters are chosen based on the values listed in TABLE \ref{tab:table1}. 

\begin{table}[h!]
\renewcommand\arraystretch{1.5}
\centering
\caption{Parameters Setting}\label{tab:table1}
\resizebox{.8\columnwidth}{!}{
\begin{tabular}{M{3.8cm}|M{4cm}}
		\hline\hline
		\textbf{\normalsize{Parameters}} & \textbf{\normalsize{Symbol and Value}}\\
		\hline\hline
		\normalsize{Carrier frequency}&\normalsize{$f_\mathrm{c}= 5$ GHz}   \\
		\hline
		\normalsize{Distance between the BS and the FRIS}&\normalsize{$d_\mathrm{BR}=200$ m}   \\
		\hline
		\normalsize{Noise power}& \normalsize{$\sigma_k^2=-110$ dBm} \\
		\hline
		\normalsize{Path-loss exponent }& \normalsize{$\alpha=2.6$}\\
		\hline
		\normalsize{Power budget at the BS }& \normalsize{$P_\mathrm{tmax}=30$ dBm}\\
		\hline
		\normalsize{Reference power gain}&\normalsize{$\rho=-20$ dB}   \\
		\hline
        \normalsize{Relative position between the BS and the FRIS}&\normalsize{$\theta_\mathrm{t, B}^\mathrm{LoS}= 70^\circ$, $\phi_\mathrm{t, B}^\mathrm{LoS}= 150^\circ$, $\theta_\mathrm{r}^\mathrm{LoS}= 110^\circ$, $\phi_\mathrm{r}^\mathrm{LoS}= -30^\circ$ }   \\ 
        \hline
        \normalsize{The number of UE}& \normalsize{$K= 3$}\\
        \hline
        \normalsize{Weighted factor}& \normalsize{$\epsilon_k= \frac{1}{3}, ~k\in\mathcal{K}$}\\
\hline
\end{tabular}
}
\end{table}

In addition, the multi-path number $L$ and $Z$ follow uniform distribution $\mathcal{U}[3, 8]$. The elevation and azimuth AoD and AoA for each path at the BS and the FRIS are independently drawn from $\mathcal{U}[10^\circ, 70^\circ]$, $\mathcal{U}[90^\circ, 150^\circ]$, $\mathcal{U}[90^\circ, 180^\circ]$ and $\mathcal{U}[-90^\circ, -30^\circ]$, respectively. The elevation and azimuth AoD of the LoS component from the FRIS to the $k$-th UE follow $\mathcal{U}[100^\circ, 110^\circ]$ and $\mathcal{U}[-20^\circ, 80^\circ]$, respectively. Also the corresponding elevation and azimuth AoD of  each path follow $\mathcal{U}[50^\circ, 160^\circ]$ and $\mathcal{U}[-70^\circ, 130^\circ]$, respectively. 

\vspace{-2mm}
\subsection{Performance Evaluation}
\subsubsection{Convergence Analysis} 
To evaluate the convergence of the proposed iterative algorithm, Fig. \ref{fig:Iterations} illustrates the simulation results of the objective function values as the number of iterations increases under different system settings. Notably, the objective function values exhibit a monotonically increasing trend with each iteration, indicating that the value at each step is no less than that of the preceding iteration throughout the iterative process. Therefore, the results confirm that the convergence of the proposed algorithm is guaranteed.

\begin{figure}[ht]
\centering
\includegraphics[width=.8\columnwidth]{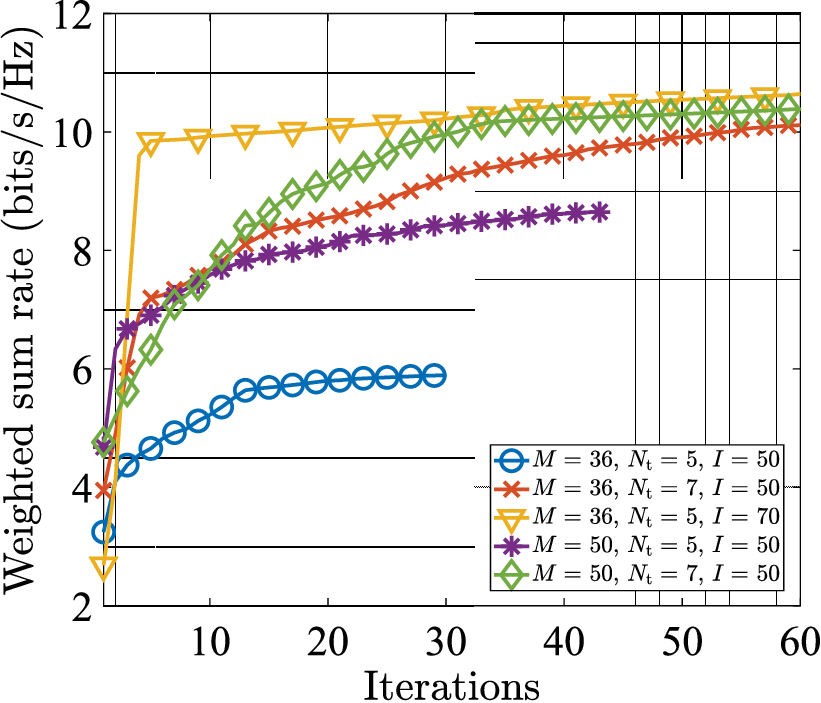}\\
\vspace{-2mm}	
\caption{The weighted sum rate against the iterations.}\label{fig:Iterations}
\vspace{-2mm}
\end{figure}

\begin{figure}[ht]
\centering
\includegraphics[width=.8\columnwidth]{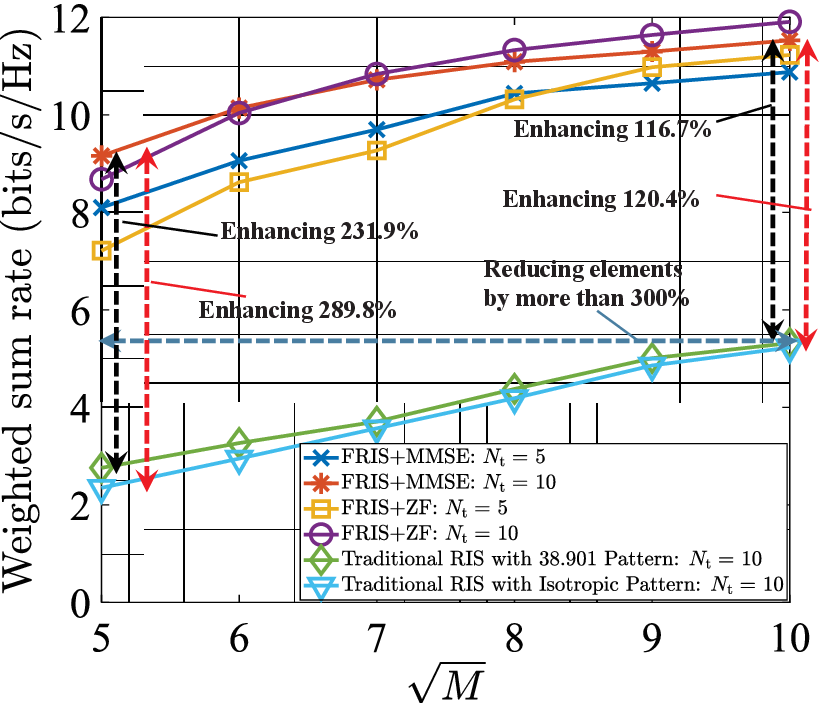}\\
\vspace{-2mm}
\caption{The weighted sum rate versus the number of elements with $I=50$.}\label{fig:WS_vs_M}
\vspace{-2mm}
\end{figure}

\subsubsection{Weighted Sum Rate versus the Number of Elements, $M$, and the Number of Antennas, $N_\mathrm{t}$} In this section, we evaluate the performance gain provided by the pattern-reconfigurable FRIS and its hardware efficiency. Fig.~\ref{fig:WS_vs_M} illustrates how the weighted sum rate varies with the number of elements deployed at the FRIS, under different antenna configurations. Specifically, we can find that the weighted sum rate is a positive function with respect to the number of elements $M$ across all cases. This is because a larger number of elements provides more DoFs to enhance the channel quality between the RIS and the BS/UEs.  An interesting trend is observed between the FRIS+MMSE and FRIS+ZF schemes. When the number of elements is small, FRIS+MMSE outperforms due to its balanced consideration of interference suppression and signal enhancement. In contrast, FRIS+ZF initially performs worse, as it focuses only on interference cancellation. However, as the number of elements increases, the enhanced signal control allows FRIS+ZF scheme to gradually surpass FRIS+MMSE. Overall, the performance gap remains small, demonstrating the effectiveness of the proposed algorithm. 

To ensure a clear performance comparison, the case of BS with $N_\mathrm{t} = 10$ antennas is utilized to implement the traditional RIS-based schemes considering 38.901 pattern and isotropic pattern. Simulation results show that even compared to the FRIS+MMSE scheme with $N_\mathrm{t} = 5$, the traditional RIS scheme using the 38.901 and isotropic radiation patterns perform significantly worse. Under the same conditions, the proposed scheme achieves improved performance gains ranging from $116.7\%$ to $231.9\%$ and $120.4\%$ to $289.8\%$ over traditional RIS with the 38.901 and isotropic patterns, respectively, as the number of elements varies from $100$ to $25$. The average performance improvements are $161.5\%$ and $176.2\%$, respectively, highlighting the strong potential of the pattern-reconfigurable FRIS in enhancing system performance. Moreover, compared to the traditional RIS scheme, the FRIS-assisted schemes can reduce the required number of elements by more than $300\%$ while maintaining the same performance (e.g., FRIS+MMSE with $M = 25$ versus traditional RIS with $M = 100$). This demonstrates the remarkable hardware efficiency of FRIS.

\begin{figure}[ht]
\centering
\includegraphics[width=.8\columnwidth]{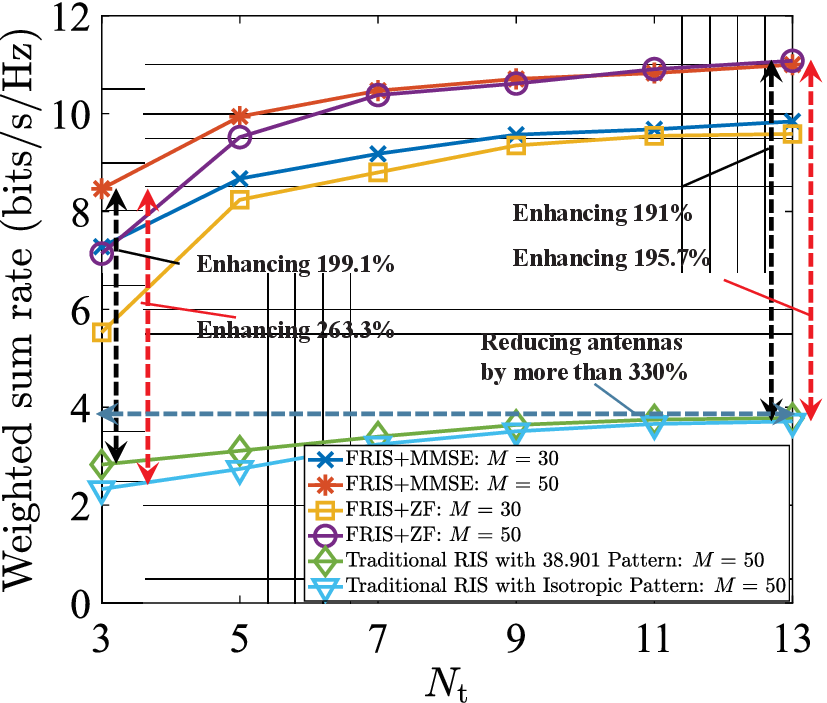}\\
\vspace{-2mm}
\caption{The weighted sum rate versus the number of antennas with $I=50$.}\label{fig:WS_vs_Nt}
\vspace{-2mm}
\end{figure}

\begin{figure*}[ht]
\centering
\includegraphics[width=.9\linewidth]{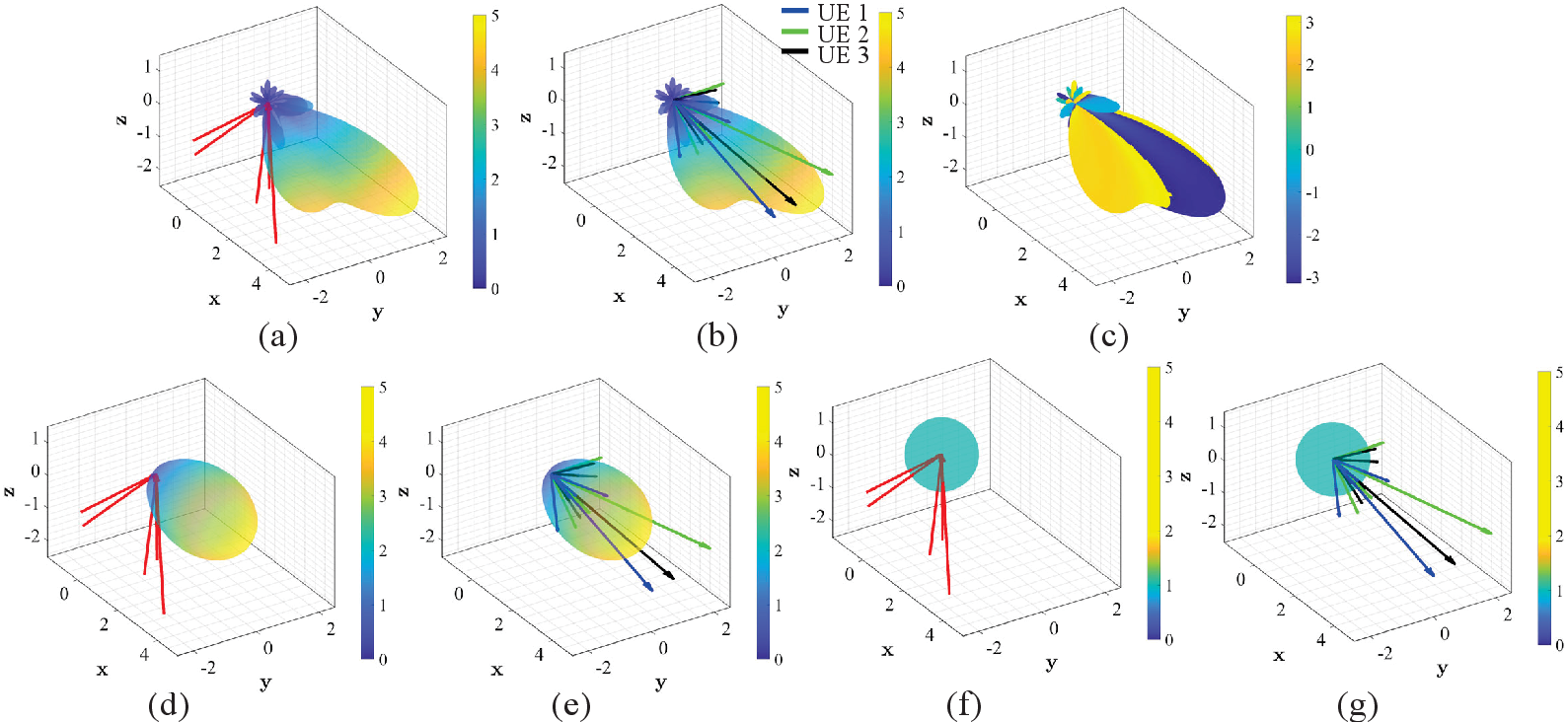}\\
\vspace{-2mm}
\caption{Visualization of the multi-path geometry and the optimal pattern of a selected element, 38.901 pattern and isotropic pattern considering $M=$ 50, $N_\mathrm{t}=$ 5, $I=$ 50: (a) Arrival multi-path geometry at the selected element and its optimal radiation pattern; (b) Departure multi-path geometry at the selected element and its optimal radiation pattern; (c) The phase distribution of the optimal radiation pattern; (d) Arrival multi-path geometry at the selected element and the 38.901 radiation pattern; (e) Departure multi-path geometry at the selected element and the 38.901 radiation pattern; (f) Arrival multi-path geometry at the selected element and the isotropic radiation pattern; (g) DAeparture multi-path geometry at the selected element and the isotropic radiation pattern. Note that the length of arrows denotes the strength of multi-path signals.}\label{fig:multi-path_pattern_illustration}
\vspace{-2mm}
\end{figure*}

The impact of the number of antennas, $N_\mathrm{t}$, on the weighted sum rate is studied under different values of $M$ in Fig.~\ref{fig:WS_vs_Nt}. As expected, the weighted sum rate exhibits an increasing trend as $N_\mathrm{t}$ increases across all scenarios. Moreover, the FRIS+MMSE scheme significantly outperforms the traditional RIS schemes with 38.901 and isotropic patterns, achieving performance enhancement ranging from $191.0\%$ to $199.1\%$ and $195.7\%$ to $263.3\%$ respectively, as the number of antennas varies from $13$ to $3$. Similarly, to achieve the same performance level, the pattern-reconfigurable FRIS can significantly reduce the required number of antennas at the BS. For example, with the assistance of the FRIS+MMSE scheme, the number of antennas can be reduced by more than $3.3$ times, as observed by comparing the FRIS+MMSE scheme with $N_\mathrm{t}=3$ and the conventional RIS-aided scheme with $N_\mathrm{t}=13$. This reduction can substantially lower the overall hardware cost.

\subsubsection{Understanding the Core Principle Through Multi-path Geometry and Pattern Alignment} Here, we aim to further understand the underlying reason why the pattern-reconfigurable FRIS achieves such significant performance gains. We focus on the optimal patterns of the elements for explanation. In Fig.~\ref{fig:multi-path_pattern_illustration}, we illustrate the spatial relationship between the multi-path geometry and the optimal pattern of a selected element, the 38.901 pattern and the isotropic pattern. As can be seen, when the BS transmits signals via a pattern-reconfigurable FRIS, the FRIS dynamically adjusts the radiation patterns of its elements based on the characteristics of both the incoming and outgoing multi-path channels. Specifically, to collect as much signal energy as possible from the BS, the radiation pattern  is diversified toward the strongest multi-path components, meaning that the main lobe is split and directed toward multiple dominant paths, while still accounting for weaker ones, as illustrated in Fig.~\ref{fig:multi-path_pattern_illustration}(a). In addition, the FRIS applies independent phase shifts to each multi-path signal to enable constructive combining at the FRIS, thereby enhancing the quality of the received signal, as shown in Fig.~\ref{fig:multi-path_pattern_illustration}(c).

In contrast, traditional RIS can only passively receive signals and applies uniform phase shifts across all the multi-path components, as observed in Fig.~\ref{fig:multi-path_pattern_illustration}(d) and Fig.~\ref{fig:multi-path_pattern_illustration}(f), thus lacking the fine-grained, path-aware pattern control enabled by FRIS. When forwarding signals to the users, the FRIS likewise adjusts its patterns to preferentially direct energy through the strongest multi-path routes while considering weaker ones, as shown in Fig.~\ref{fig:multi-path_pattern_illustration}(b). Independent phase control is applied to ensure coherent combining at the user side.

\begin{figure}[ht]
\centering
\includegraphics[width=.8\columnwidth]{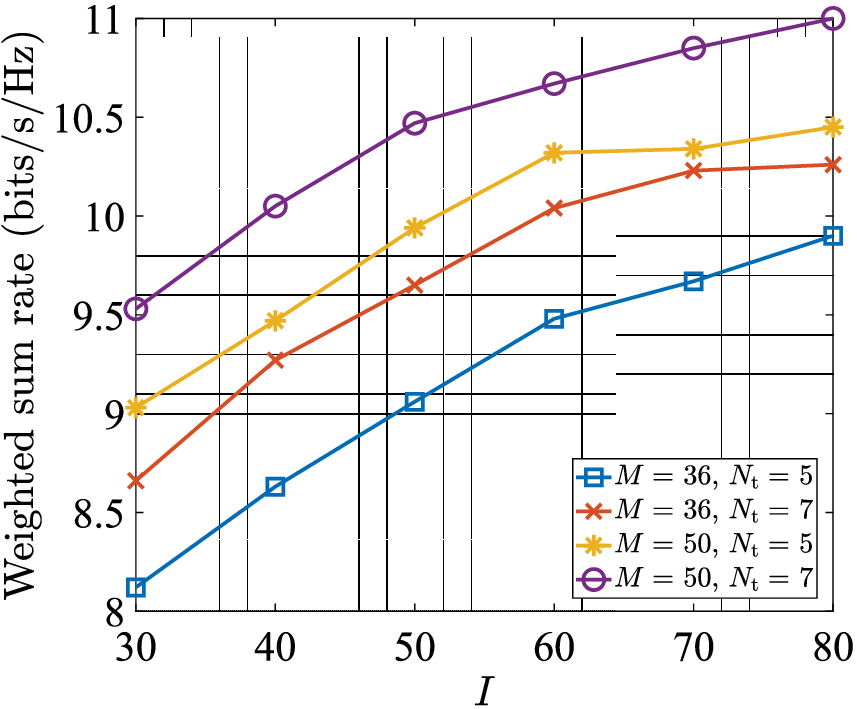}\\
\vspace{-2mm}
\caption{The weighted sum rate versus the truncated length $I$.}\label{fig:WS_vs_I}
\vspace{-2mm}
\end{figure}

\begin{figure}[ht]
\centering
\includegraphics[width=.8\columnwidth]{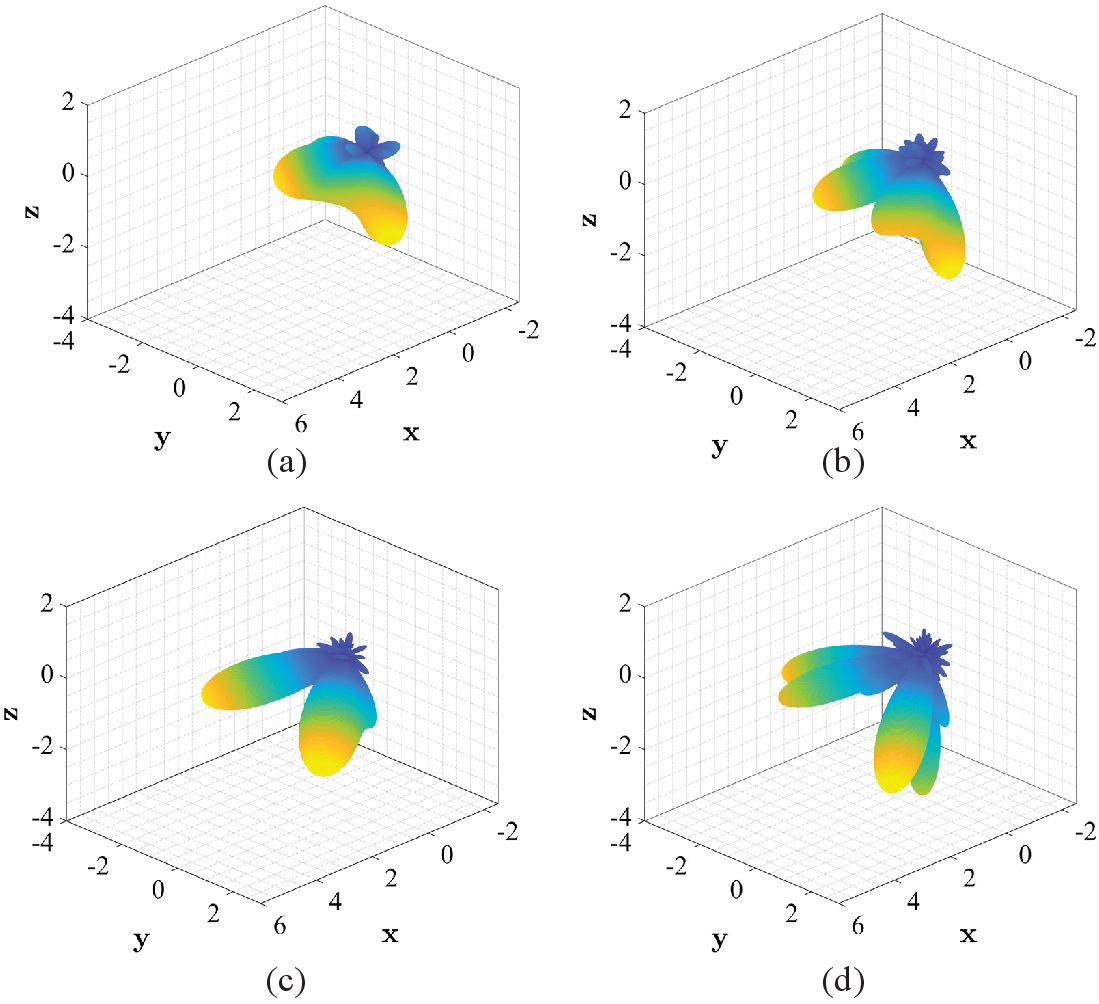}\\
\vspace{-2mm}
\caption{The optimal pattern versus the the truncated length $I$ with (a) $I=25$; (b) $I=36$; (c) $I=49$; (d) $I=81$.}\label{fig:pattern_vs_I}
\vspace{-2mm}
\end{figure}

\subsubsection{Weighted Sum Rate versus Truncated Length $I$} Now, we conclude this section by investigating the varied correlation between the weighted sum rate and the truncated length $I$, taking into account the different number of elements and antennas in Fig.~\ref{fig:WS_vs_I}. It is observed that all the cases indicate a rising trend as $I$ increases, while the rate of growth slows down. This trend can be explained by examining the optimal radiation pattern of a selected element under different values of $I$, as illustrated in Fig.~\ref{fig:pattern_vs_I}. In particular, we see that as $I$ increases, the DoFs for pattern reconfiguration expand, resulting in more directive main lobes in each element's radiation pattern. This enhances the capability to capture and deliver more signal energy to the users, thereby leading to a steady increase in the weighted sum rate with larger values of $I$. The diminishing rate of increase in the weighted sum rate is primarily due to the fact that the energy of the radiation pattern is mainly concentrated in the low-frequency components (as evidenced by the results in Fig.~\ref{fig:SHOD_method_verification}). As a result, the effective increase in the DoFs of the pattern reconfiguration becomes marginal as $I$ grows, leading to a more gradual improvement in performance.

\vspace{-2mm}
\section{Conclusion}\label{sec:S6}
This paper proposed a novel pattern-reconfigurable FRIS framework. We first investigated the potential by comparing the received signal power in communication systems aided by the pattern-reconfigurable FRIS, the position-reconfigurable FRIS, and the traditional RIS. Theoretical results showed that the pattern-reconfigurable FRIS offers significant advantages in modulating transmission signals over the other two. Motivated by this, we then considered a multiuser communication scenario aided by a pattern-reconfigurable FRIS. The SHOD method was employed to accurately model the fluid elements' patterns, and we subsequently formulated an optimization problem to maximize the weighted sum rate, subject to power and pattern energy constraints. An iterative algorithm based on the MMSE and RCG methods was developed. Simulation results illustrated that the proposed pattern-reconfigurable FRIS delivers remarkable performance improvements compared to traditional RIS using 3GPP 38.901 and isotropic patterns. 

\appendices

\ifCLASSOPTIONcaptionsoff 
  \newpage
\fi
\bibliographystyle{IEEEtran}

\end{document}